\newif\ifWP
\WPfalse
\newif\ifJOURNAL
\JOURNALfalse

\newif\ifFULL
\FULLfalse
\newif\ifLATIN
\LATINfalse

\WPtrue		

\LATINtrue	

\newif\ifnotFULL	
\notFULLtrue
\ifFULL\notFULLfalse\fi

\newif\ifnotLATIN	
\notLATINtrue
\ifLATIN\notLATINfalse\fi

\ifWP
  \newcommand{\GTPxxviii}{GTP28arXiv}
  \newcommand{\GTPxxix}{GTP29}
  \newcommand{\GTPxxxiii}{GTP33-full}
\fi
\ifJOURNAL
  \newcommand{\GTPxxviii}{vovk:FS}
  \newcommand{\GTPxxix}{shafer/etal:arXiv0905}
  \newcommand{\GTPxxxiii}{shafer/etal:2011}
\fi

\ifnotLATIN
  
\fi
\ifLATIN
  
\fi

\ifWP
  \documentclass[10pt]{article}
  \usepackage{amsmath,amsfonts,amssymb,amsthm,latexsym,graphicx,algorithm,algorithmic}

\makeatletter

\newif\iftwodates
\twodatesfalse

\renewcommand\maketitle{\begin{titlepage}%
  \let\footnotesize\small
  \let\footnoterule\relax
  \let \footnote \thanks
  \null\vfil
  \vskip 30\p@
  \begin{center}%
    {\LARGE \bf \@title \par}%
    \vskip 3em%
    {\large
     \lineskip .75em%
     \begin{tabular}[t]{c}%
       \@author
     \end{tabular}\par}%
     \vskip 1.5em%
  \end{center}\par
  \vfill
  \begin{center}
    \raisebox{1.5cm}{\includegraphics[width=0.58\textwidth]%
      {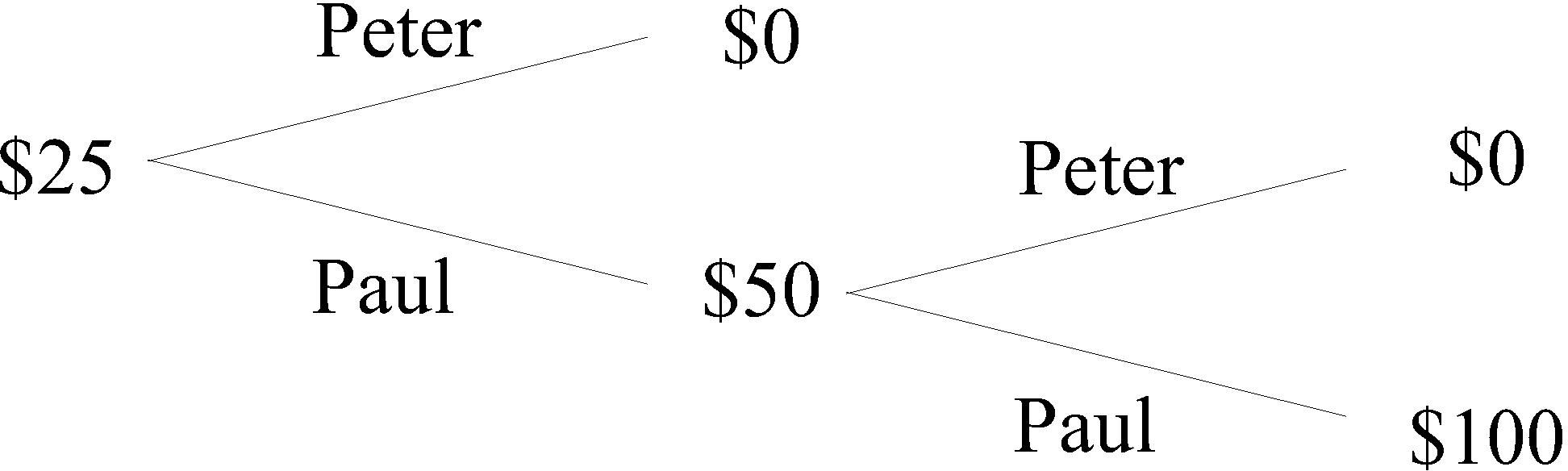}}%
    \hskip 3em%
    \includegraphics[width=0.29\textwidth]%
      {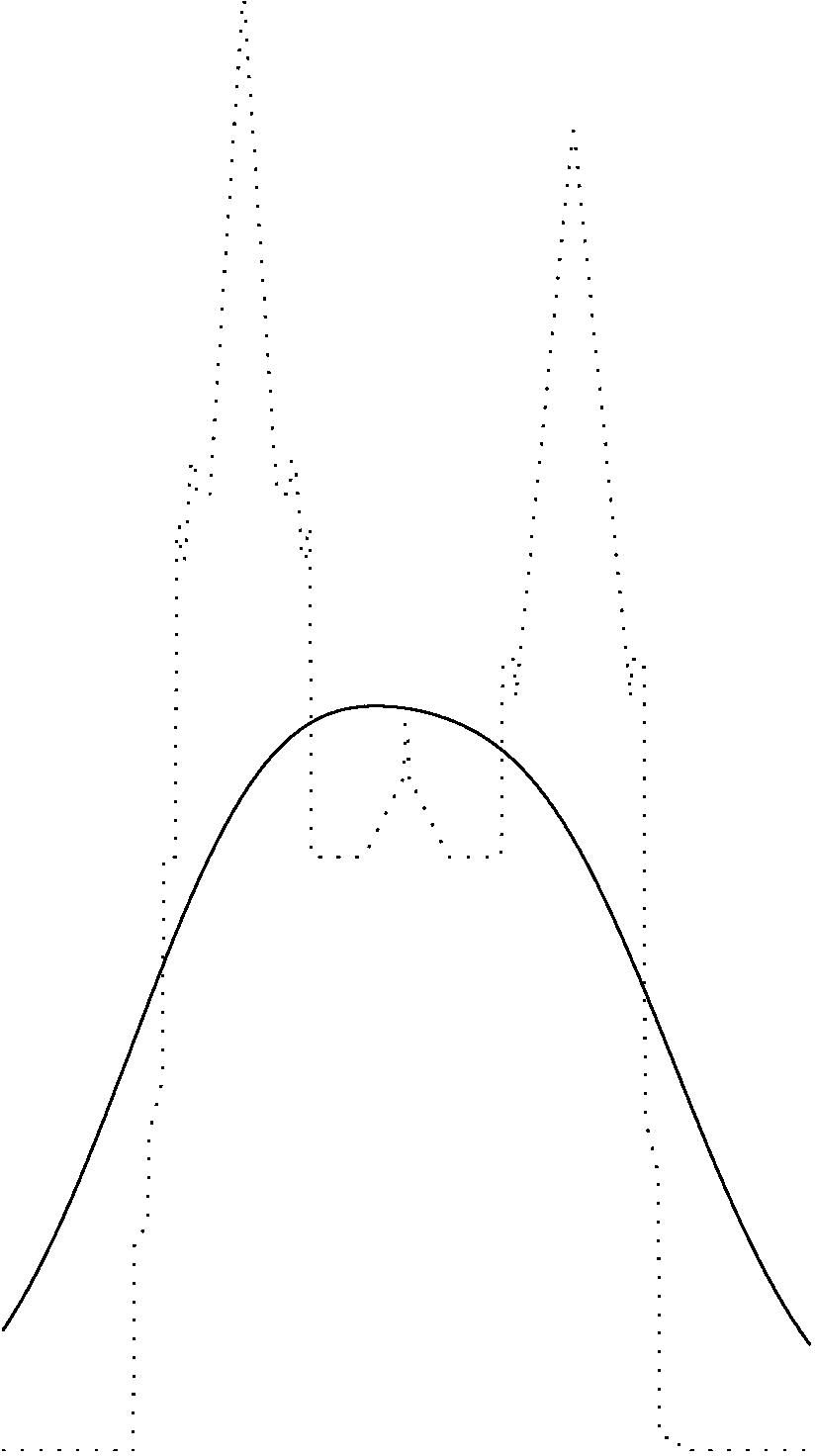}%
  \end{center}
  \@thanks
  \vfill
  \begin{center}
    {\large \bf The Game-Theoretic Probability and Finance Project}
  \end{center}
  \begin{center}
    {\large Working Paper \#\No}
  \end{center}
  \begin{center}
    {\iftwodates\large First posted \firstposted.
    Last revised \@date.\else\large\@date\fi}
  \end{center}
  \begin{center}
    Project web site:\\
    http://www.probabilityandfinance.com
  \end{center}
  \end{titlepage}%
  \setcounter{footnote}{0}%
  \global\let\thanks\relax
  \global\let\maketitle\relax
  \global\let\@thanks\@empty
  \global\let\@author\@empty
  \global\let\@date\@empty
  \global\let\@title\@empty
  \global\let\title\relax
  \global\let\author\relax
  \global\let\date\relax
  \global\let\and\relax
}

\renewenvironment{abstract}{%
  \titlepage
  \null\vfil
  \@beginparpenalty\@lowpenalty
  \begin{center}%
    \Large \bfseries \abstractname
    \@endparpenalty\@M
  \end{center}}%
  {\par\vfill\tableofcontents\endtitlepage}

\renewenvironment{thebibliography}[1]
  {\section*{\refname}%
  \addcontentsline{toc}{section}{\refname}
  \@mkboth{\MakeUppercase\refname}{\MakeUppercase\refname}%
  \list{\@biblabel{\@arabic\c@enumiv}}%
    {\settowidth\labelwidth{\@biblabel{#1}}%
    \leftmargin\labelwidth
    \advance\leftmargin\labelsep
    \@openbib@code
    \usecounter{enumiv}%
    \let\p@enumiv\@empty
    \renewcommand\theenumiv{\@arabic\c@enumiv}}%
    \sloppy
    \clubpenalty4000
    \@clubpenalty \clubpenalty
    \widowpenalty4000%
    \sfcode`\.\@m}
    {\def\@noitemerr
    {\@latex@warning{Empty `thebibliography' environment}}%
  \endlist}

\makeatother

\fi

\ifJOURNAL
  \documentclass[12pt,leqno]{article}
  \usepackage{amsmath,amsfonts,amssymb,amsthm,latexsym,graphicx,algorithm,algorithmic}
\fi

\ifFULL
  \usepackage{color}

  \newcommand{\bluebegin}{\begingroup\color{blue}}
  \newcommand{\blueend}{\endgroup}

\fi

\allowdisplaybreaks[4]

\newcommand{\Vladimir}{Vladimir}
\newcommand{\DOT}{.}

\ifnotLATIN
  \input{OT2enc.def}
  
  \usepackage{CJK}
\fi

\gdef\mathenddots{\mathinner{\ldotp\ldotp\ldotp\ldotp}}

\newcommand{\st}{\mathrel{|}}
\newcommand{\givn}{\mathrel{|}}

\newcommand{\dd}{\mathrm{d}}		

\newcommand{\K}{\mathcal{K}}		
\newcommand{\CCC}{\mathcal{C}}		
\newcommand{\EEE}{\mathcal{E}}		

\newcommand{\rd}{_{\rm r}}		

\newcommand{\X}{\mathbf{X}}		
\newcommand{\Xstar}{\mathbf{X^*}}	

\DeclareMathOperator{\III}{\boldsymbol{1}}	

\newcommand{\bbbn}{\mathbb{N}}                  
\newcommand{\bbbr}{\mathbb{R}}			

\newcommand{\bbbe}{\mathbb{E}}		
\DeclareMathOperator{\Expect}{\bbbe}

\DeclareMathOperator{\UpExpect}{\lefteqn{\smash{\overline{\bbbe}}}\phantom{\bbbe}}

\newcommand{\XXX}{\mathcal{X}}		

\newcommand{\bbbrbar}{\overline{\mathbb{R}}}	

\newcommand{\bbbrbarXXX}{\smash{\bbbrbar}\vphantom{\bbbr}^{\XXX}}

\theoremstyle{plain}
\newtheorem{theorem}{Theorem}[section]
\newtheorem{proposition}[theorem]{Proposition}
\newtheorem{corollary}[theorem]{Corollary}
\newtheorem{lemma}[theorem]{Lemma}

\theoremstyle{definition}
\newtheorem*{remark}{Remark}

\makeatletter
  \@addtoreset{equation}{section}
  
\makeatother

\newcounter{protocol}
\newcounter{temporary}
\makeatletter
  \newenvironment{protocol}{%
  \renewcommand{\ALG@name}{Protocol}
  \setcounter{temporary}{\value{algorithm}}
  \setcounter{algorithm}{\value{protocol}}
  \begin{algorithm}
  }{%
  \end{algorithm}
  \renewcommand{\ALG@name}{Algorithm}
  \setcounter{protocol}{\value{algorithm}}
  \setcounter{algorithm}{\value{temporary}}
  }
\makeatother

\title{Probability-free pricing of adjusted American lookbacks}
\author{A. Philip Dawid, Steven de Rooij, Peter Gr\"unwald, Wouter M. Koolen,\\
   Glenn Shafer, Alexander Shen, Nikolai Vereshchagin, and Vladimir Vovk}
  \newcommand{\No}{37}

\newcommand{\acknowledge}{We are grateful to the participants
  of the discussion of our work on Alexei Savvateev's blog
  (\texttt{savvateev.livejournal.com/66517.html})
  and to Andrzej Ruszczy\'nski's for his advice.
  This research has been supported in part
  by ANR grant NAFIT ANR-08-EMER-008-01, EPSRC grant EP/F002998/1,
  and NWO Rubicon grant 680-50-1010.}

\begin{document}
\ifWP
  \maketitle
\fi

\ifJOURNAL
  \begin{titlepage}
  
  \vfill

  \begin{center}
    {\huge Probability-free pricing of adjusted American lookbacks}{\Large \footnotemark[1]}
  \end{center}

  \medskip

    \noindent
    \textbf{A. Philip Dawid} (University of Cambridge, UK),
    \textbf{Steven de Rooij} (Centrum Wiskunde \&\ Informatica, Amsterdam, Netherlands),
    \textbf{Peter Gr\"unwald} (Centrum Wiskunde \&\ Informatica, Amsterdam, Netherlands),
    \textbf{Wouter M. Koolen} (Royal Holloway, University of London, UK,
      and Centrum Wiskunde \&\ Informatica, Amsterdam, Netherlands),
    \textbf{Glenn Shafer} (Rutgers Business School, Newark, NJ, USA),
    \textbf{Alexander Shen} (Laboratoire d'Informatique Fondamentale, CNRS, Marseille, France),
    \textbf{Nikolai Vereshchagin} (Department of Mathematical Logic and the Theory of Algorithms,
      Moscow State University, Moscow, Russia),
    and \textbf{Vladimir Vovk} (Royal Holloway, University of London, UK)

  \bigskip


  \bigskip

  \begin{center}\textbf{\Large Abstract}\end{center}

\fi

\ifWP
  \begin{abstract}
\fi
  Consider an American option that pays $G(X^*_t)$ when exercised at time $t$,
  where $G$ is a positive increasing function, $X^*_t:=\sup_{s\le t}X_s$,
  and $X_s$ is the price of the underlying security at time $s$.
  Assuming zero interest rates,
  we show that the seller of this option can hedge his position
  by trading in the underlying security
  if he begins with initial capital
  $X_0\int_{X_0}^{\infty}G(x)x^{-2}\dd x$
  (and this is the smallest initial capital that allows him to hedge his position).
  This leads to strategies for trading that are always
  competitive both with a given strategy's current performance
  and, to a somewhat lesser degree,
  with its best performance so far.
  It also leads to methods of statistical testing
  that avoid sacrificing too much of the maximum statistical significance
  that they achieve in the course of accumulating data.
  %
  %
\ifWP
  \end{abstract}
\fi

\ifJOURNAL

  \vfill

  \footnotetext[1]{\acknowledge}
  \end{titlepage}
\fi

\section{Introduction}
\label{sec:introduction}

A financial security, such as a stock,
that gains in price for a period of time may do much worse later
losing much or all of its value.
When this happens,
an investor who persisted in holding the security
will regret not having sold it when its price was high.
This motivates lookback options,
which permit the investor to claim the maximum price
attained over a period of time.
Established methods for pricing such claims
either depend on probabilistic assumptions
about the behaviour of the security price
or assume that some other derivatives, such as call options,
are priced by the market and are available for trading.
In this article,
we show that when only a reasonable fraction of the maximum price is demanded,
such claims can be priced and hedged without any probabilistic assumptions
and without relying on any other derivatives.

Let $X_t$ be the security's price at time $t$
and set
$$ 
  X^*_t:=\sup_{s\le t}X_s.
$$
The classic American lookback option has $X^*_t$ as its payoff
when it is exercised at time $t$.
We explain how to find probability-free upper prices
for more general American options,
options that pay $G(X^*_t,X_t)$ when exercised at time $t$,
where $G$ is a given positive function of two variables.
We call such an option an adjusted American lookback option.
The least initial capital needed to finance a trading strategy
whose capital $\K_t$ will satisfy $\K_t \ge G(X^*_t,X_t)$ for all $t$
regardless of how the prices of the underlying securities evolve
is called the option's \emph{upper price}.
This term is standard in game-theoretic probability \cite{shafer/vovk:2001}.
The upper price of an option is what a seller needs
in order to hedge fully against possible loss,
while the lower price is what a buyer needs for the same purpose.
In an incomplete market the two are not necessarily equal,
and since we are assuming neither probabilities nor market pricing of other options,
our market is very incomplete.
To emphasize that we make no probabilistic assumptions
about the underlying security prices,
we sometimes call our upper prices probability-free.

A closely related and conceptually simpler problem
is whether there is a strategy for the investor
that keeps its capital greater than or equal to $F(X^*_t,X_t)$,
where $F$ is a given positive function of two variables.
For simplicity (but without loss of generality)
we will always consider this question with the initial price $X_0$ fixed to 1.
If there is a strategy whose capital process $\K_t$ satisfies $\K_0=1$
and $\K_t \ge F(X^*_t,X_t)$ for all $t$, and for any price evolution from $X_0=1$,
then we call $F$ a \emph{lookback adjuster}, or \emph{LA}.
If $F$ is an LA and there is no other LA that dominates it,
we call $F$ an \emph{admissible lookback adjuster}, or \emph{ALA}.
We show that every LA is dominated by an ALA, and we characterize the ALAs
(Theorem~\ref{thm:general}).

The picture is clearest in the case of adjusters
and options that depend only on $X^*_t$ (i.e., not on $X_t$).
We call these options and adjusters simple.
\begin{itemize}
\item
  \textbf{Simple lookback adjusters.}
  We call an increasing right-continuous positive function $F$ of one variable
  a \emph{simple lookback adjuster}, or \emph{SLA},
  if there is a strategy for the investor that starts with initial capital $1$
  and keeps its capital greater than or equal to $F(X^*_t)$ for all $t\ge 0$,
  on the assumption that $X_0=1$.
  If $F$ is an SLA and there is no other SLA that dominates it,
  we call $F$ an \emph{admissible simple lookback adjuster}, or \emph{ASLA}.
  We show that $F$ is an SLA if and only if 
  \begin{equation}\label{eq:SLA}
    \int_1^{\infty}
    \frac{F(y)}{y^2}
    \dd y
    \le
    1,
  \end{equation}
  and that $F$ is an ASLA if and only if~(\ref{eq:SLA}) holds with equality
  (Proposition~\ref{prop:SLA}).
\item
  \textbf{Simple lookback options.}
  Consider an American option 
  that pays $G(X^*_t)$ if exercised at time $t$,
  where $G$ is a given increasing right-continuous positive function;
  we only assume that $X_0>0$.
  It follows from the criterion (\ref{eq:SLA}) that this option's 
  upper price at time $0$
  is
  \begin{equation}\label{eq:upper-price}
    X_0
    \int_{X_0}^{\infty}
    \frac{G(x)}{x^{2}}
    \dd x.
  \end{equation}
  Indeed, applying a strategy always ensuring capital $\K_t\ge F(Y^*_t)$
  to the normalised price $Y_t:=X_t/X_0$ (which satisfies $Y_0=1$) and to
  $$
    F(y)
    :=
    \frac{G(X_0y)}{X_0\int_{X_0}^{\infty}G(x)x^{-2}\dd x},
    \quad
    y\in[1,\infty),
  $$
  (satisfying (\ref{eq:SLA}) if we exclude the trivial case $G(x)=0$, $\forall x\ge X_0$),
  we can ensure that
  $$
    \K_t
    \ge
    F(Y^*_t)
    =
    F(X^*_t/X_0)
    =
    \frac{G(X^*_t)}{X_0\int_{X_0}^{\infty}G(x)x^{-2}\dd x}
  $$
  with initial capital $1$;
  therefore, we can ensure that our capital is always
  at least $G(X^*_t)$ with initial capital
  $X_0\int_{X_0}^{\infty}G(x)x^{-2}\dd x$
  (but not with less).
\end{itemize}

The left-hand side of (\ref{eq:SLA})
is the expected value of $F(y)$
when $y$ follows the probability measure $Q_1$ on $[1,\infty)$
whose density is $y^{-2}$.
More generally,
(\ref{eq:upper-price}) is the expected value of $G(x)$
when $x$ follows the probability measure $Q_{X_0}$ on $[X_0,\infty)$
whose density is $X_0x^{-2}$.
This conforms to the standard picture
in which option prices are expected values
with respect to probability distributions,
conventionally called ``risk-neutral'',
which emerge naturally instead of being assumed
(in the case of (\ref{eq:SLA}),
the ``option price'' is the initial unit capital).
What is unusual here is that the risk-neutral measures emerge
even in a heavily incomplete market.
The measure $Q_{X_0}$ is the distribution of the maximum of Brownian motion
started at $X_0$ and stopped when it hits 0;
we will examine this connection further in Section~\ref{sec:literature}.



The most basic lookback option $G(X^*)$ is simply $X^*$,
paying $X^*_t$ at a time $t$ of the owner's choice.
Its upper price is infinite:
$X_0\int_{X_0}^{\infty}x^{-1}\dd x=\infty$.
To get a finite upper price,
we can fix a finite maturity date $T$
and consider the European lookback option with payoff $X_T^*$.
Hobson \cite{hobson:1998} derives upper prices for options of this type
on the assumption that the market prices 
call options on $X$ with maturity date $T$
and all possible strike prices.
Hobson's work has been developed in various directions:
see, e.g., the recent review \cite{hobson:2011}
and references therein.  We are not aware, however, of work on lookbacks
that relies neither on probabilistic assumptions
nor on market pricing of other options.
For other connections with existing literature,
see Section~\ref{sec:literature}.

The centrepiece of this article is Figure~\ref{fig:frame},
which establishes connections between several seemingly very different notions.
Part of this study has been published as \cite{dawid/etal:2011}
in \emph{Statistics and Probability Letters}.

\subsection*{Terminology, notation, and abbreviations}

We use terms such as ``positive'', ``increasing'', and ``above''
in the wide sense of the inequalities $\le$ and $\ge$.
We use the standard symbol $\bbbr$ for the set of real numbers;
the set of natural numbers is $\bbbn:=\{1,2,\ldots\}$.
We never use primes to mean differentiation;
instead, we use the more specific notation $f\rd$
to mean the right derivative of $f$
(we will use it mainly for concave functions $f$,
when $f\rd$ is guaranteed to exist).
In Section~\ref{sec:GTP},
the extended real line $[-\infty,\infty]$ will be denoted $\bbbrbar$,
and we will use the convention $\infty+(-\infty):=\infty$.

This is the list of abbreviations used in this article:
\begin{description}
\item[ALA]
  admissible lookback adjuster,
  often denoted $F(X^*,X)$
\item[ASLA]
  admissible simple lookback adjuster,
  often denoted $F(X^*)$
\item[LA]
  lookback adjuster,
  often denoted $F(X^*,X)$
\item[SLA]
  simple lookback adjuster,
  often denoted $F(X^*)$
\end{description}

\section{Insuring against loss of capital, I}
\label{sec:capital-1}

This section's (and most of this article's) trading protocol
is given as Protocol~\ref{prot:competitive-trading}.
It describes a perfect-information game between two players,
Market and Investor.
The players make their moves sequentially in the indicated order.
There is one security, often referred to as $X$,
whose price $X_t$ at time $t>0$ is chosen by Market.
We will refer to $p_t$ as Investor's position in $X$ at time $t$,
or the number of units of $X$ that he holds at time $t$.
For simplicity,
the protocol and our formal results
cover only the case of discrete time,
although in our informal discussions
we will sometimes consider the case of continuous time,
$t\in[0,\infty)$.

\begin{protocol}
  \caption{Simplified trading in a financial security}
  \label{prot:competitive-trading}
  \begin{algorithmic}
    \STATE $X_0:=1$ and $\K_0:=1$
    \FOR{$t=1,2,\dots$}
      \STATE Investor announces $p_t\in\bbbr$
      \STATE Market announces $X_t\in[0,\infty)$
      \STATE $\K_t:=\K_{t-1}+p_t(X_t-X_{t-1})$
    \ENDFOR
  \end{algorithmic}
\end{protocol}

In the bulk of the article we will consider the conceptually simplest case of one security $X$.
However, we may always think of $X_t$ as the capital of a trading strategy, fund, or adviser
when trading in a multi-security market.

In terms of Protocol~\ref{prot:competitive-trading},
we call an increasing function
$F:[1,\infty)\to[0,\infty)$ an \emph{SLA}
if there exists a strategy for Investor that guarantees
$\K_t\ge F(X_t^*)$ for all $t$.
We say that an SLA $F$ \emph{dominates} an SLA $G$
if $F(y)\ge G(y)$ for all $y\in[1,\infty)$.
We say that $F$ \emph{strictly dominates} $G$
if $F$ dominates $G$ and $F(y)>G(y)$ for some $y\in[1,\infty)$.
An SLA is an \emph{ASLA} if it is not strictly dominated by any SLA.

\begin{proposition}\label{prop:SLA}
\begin{enumerate}
\item
  An increasing function $F:[1,\infty)\to[0,\infty)$ is an SLA
  if and only if it satisfies (\ref{eq:SLA}).
\item
  Any SLA is dominated by an ASLA.
\item
  An SLA is admissible (is an ASLA) if and only if
  it is right-continuous and
  \begin{equation}\label{eq:ASLA}
    \int_1^{\infty}
    \frac{F(y)}{y^2}
    \dd y
    =
    1.
  \end{equation}
\end{enumerate}
\end{proposition}

We will give two proofs of this result:
in this section we will give a simple direct derivation,
and in Section~\ref{sec:connections}
we will derive it from a much more general statement.

The main idea of the direct derivation is as follows.
For every threshold $u$ we consider the strategy
that holds 1 unit of $X$, selling it when Investor's capital reaches (or exceeds) $u$.
This corresponds to the SLA $F_u(y):=u\III_{\{y\ge u\}}$.
(If $E$ is some property,
$\III_{\{E\}}$ is defined to be $1$ if $E$ is satisfied
and $0$ if not.)
Now we can mix these strategies
according to some probability measure $P$ on $u$.
It remains to notice that every increasing function $F$
satisfying (\ref{eq:SLA})
can be represented as such a mixture:
$F(y)=\int_1^{\infty} F_u(y) P(\dd u) = \int_1^y u P(\dd u)$.
Now we give a formal proof of part of Proposition~\ref{prop:SLA}
and an informal argument for the remaining part.

\begin{proof}[Proof of Proposition~\ref{prop:SLA}]
  First we prove that any increasing function $F:[1,\infty)\to[0,\infty)$ satisfying
  \begin{equation}\label{eq:measure2SLA}
    F(y)
    =
    \int_{[1,y]}
    u
    P(\dd u),
    \quad
    \forall y\in[1,\infty),
  \end{equation}
  for a probability measure $P$ on $[1,\infty]$,
  is an SLA.
  For each $u\ge1$, define the following strategy for Investor:
  on round $t$, the strategy outputs
  \begin{equation}\label{eq:strategy}
    p^{(u)}_t
    :=
    \begin{cases}
      1 & \text{if $X^*_{t-1}<u$}\\
      0 & \text{otherwise}
    \end{cases}
  \end{equation}
  as Investor's move $p_t$.
  (Intuitively, this strategy holds 1 unit of $X$ until $X$'s price reaches $u$;
  as soon as this happens, $X$ is sold.)
  Let $\K_t^{(u)}$ be the capital process of this strategy.
  Set
  \begin{equation}\label{eq:p}
    p_t:=\int_{[1,\infty]}p_t^{(u)}P(\dd u).
  \end{equation}
  This gives $\K_t=\int_{[1,\infty]} \K_t^{(u)} P(\dd u)$:
  indeed, this is true for $t=0$ and the inductive step is
  \begin{align*}
    \K_t
    &=
    \K_{t-1} + p_t(X_t-X_{t-1})\\
    &=
    \int_{[1,\infty]} \K_{t-1}^{(u)} P(\dd u)
    +
    \int_{[1,\infty]}p_t^{(u)}P(\dd u)
    (X_t-X_{t-1})\\
    &=
    \int_{[1,\infty]}
    \left(
      \K_{t-1}^{(u)}
      +
      p_t^{(u)}
      (X_t-X_{t-1})
    \right)
    P(\dd u)\\
    &=
    \int_{[1,\infty]} \K_t^{(u)}P(\dd u).
  \end{align*}
  This strategy will guarantee
  \begin{equation}\label{eq:guarantee}
    \K_t
    =
    \int_{[1,\infty]} \K_t^{(u)} P(\dd u)
    \ge
    \int_{[1,X_t^*]} \K_t^{(u)} P(\dd u)
    \ge
    \int_{[1,X_t^*]} u P(\dd u)
    =
    F(X_t^*).
  \end{equation}

  We can now finish the proof of the statement ``if'' in part~1 of the proposition,
  which says that any increasing function $F:[1,\infty)\to[0,\infty)$
  satisfying~(\ref{eq:SLA}) is an SLA.
  Without loss of generality we can 
  assume that $F$ is right-continuous and that (\ref{eq:ASLA}) holds.
  It remains to apply Lemma~\ref{lem:ASLA} below.

  Let us now check that every SLA satisfies~(\ref{eq:SLA}).
  Our argument will be informal:
  first, it is easy to formalize,
  and second, in Section~\ref{sec:connections} we will deduce this statement
  independently (see Corollary~\ref{cor:SLA}).
  Consider the case of continuous time,
  where the security price $X_t$ depends on $t\in[0,\infty)$
  and Investor's capital $\K_t$ is defined as in \cite{\GTPxxviii}, (2).
  Investor can guarantee $\K_t\ge F(X_t^*)$, $\forall t$.
  Let $X_t$ be the trajectory of Brownian motion started at 1
  and stopped when it hits 0 for the first time.
  The distribution of $X^*_{\infty}$ has density $y^{-2}$, $y\in[1,\infty)$
  (see Section~\ref{sec:literature} for details).
  The expected value of $F(X_{\infty}^*)$ is equal
  to the left-hand side of (\ref{eq:SLA}).
  Since $\K_t$ is a positive supermartingale with initial value 1,
  we obtain that the left-hand side of (\ref{eq:SLA})
  does not exceed
  $$
    \Expect \liminf_{t\to\infty} \K_t
    \le
    \liminf_{t\to\infty} \Expect \K_t
    \le
    1.
  $$
  To formalize this argument,
  it suffices to replace the Brownian motion with the random walk started from 1
  with the increment $\pm1/N$ for a large $N$
  (the $\pm$ is $+$ or $-$ with probability $1/2$).

  We have established part~1 of the theorem.
  Part~3 is now obvious, and part~2 follows from parts~1 and~3.
\end{proof}

The method used in this proof (stopping and combining)
has been used previously by various authors, e.g.,
El-Yaniv et al.\
(\cite{el-yaniv/etal:2001}, Theorem 1,
based on Leonid Levin's personal communication)
and Shafer and Vovk
(\cite{shafer/vovk:2001}, Lemma 3.1).
We have now seen that it gives optimal results in our setting.

The second statement of the following lemma
was used in the proof of Proposition~\ref{prop:SLA}.
\begin{lemma}\label{lem:ASLA}
  An increasing right-continuous function $F:[1,\infty)\to [0,\infty)$
  satisfies~(\ref{eq:ASLA}) if and only if (\ref{eq:measure2SLA})
  holds for some probability measure $P$ on $[1,\infty)$.
  It satisfies~(\ref{eq:SLA}) if and only if (\ref{eq:measure2SLA})
  holds for some probability measure $P$ on $[1,\infty]$.
\end{lemma}
\begin{proof}
  It is sufficient to prove the first statement of the lemma;
  the second then follows easily.

  Let us first check that the existence of a probability measure $P$
  on $[0,\infty)$
  satisfying~(\ref{eq:measure2SLA})
  implies~(\ref{eq:ASLA}).
  We have:
  \begin{multline}
    \int_{[1,\infty)}
    \frac{F(y)}{y^2}
    \dd y
    =
    \int_{[1,\infty)}
    \int_{[1,y]}
    \frac{u}{y^2}
    P(\dd u)
    \dd y \\
    =
    \int_{[1,\infty)}
    \int_{[u,\infty)}
    \frac{u}{y^2}
    \dd y
    P(\dd u)
    =
    \int_{[1,\infty)}
    P(\dd u)
    =
    1.
    \label{eq:integral}
  \end{multline}

  It remains to check that any increasing right-continuous
  $F:[1,\infty)\to[0,\infty)$
  satisfying~(\ref{eq:ASLA}) satisfies~(\ref{eq:measure2SLA})
  for some probability measure $P$ on $[1,\infty)$.
  Let $Q$ be the measure on $[1,\infty)$
  ($\sigma$-finite but not necessarily a probability measure)
  with distribution function $F$,
  in the sense that $Q([1,y])=F(y)$
  for all $y\in[1,\infty)$.
  Set $P(\dd u):=(1/u)Q(\dd u)$.
  We then have~(\ref{eq:measure2SLA}),
  and the calculation (\ref{eq:integral})
  shows that the $\sigma$-finite measure $P$ must be a probability measure
  (were it not, we would not have an equality in~(\ref{eq:ASLA})).
\end{proof}

According to~(\ref{eq:ASLA}),
the function
\begin{equation}\label{eq:class-1}
  F(y)
  :=
  \alpha y^{1-\alpha}
\end{equation}
is an ASLA for any $\alpha\in(0,1)$
(\cite{\GTPxxxiii}, (12)).
Another example (\cite{\GTPxxxiii}, below (12)) is
\begin{equation}\label{eq:class-2}
  F(y)
  :=
  \begin{cases}
    \alpha (1+\alpha)^{\alpha} \frac{y}{\ln^{1+\alpha}y}
      & \text{if $y \ge e^{1+\alpha}$}\\
    0  & \text{otherwise},
  \end{cases}
\end{equation}
where $\alpha>0$.
The measures $P$ corresponding (see (\ref{eq:measure2SLA}))
to (\ref{eq:class-1}) and (\ref{eq:class-2})
are computed in Appendix \ref{app:analytic}.

\section{Insuring against loss of capital, II}
\label{sec:capital-2}

The previous section explains how we can get an insurance
against losing almost all capital
as compared to the peak price of the underlying security.
In this section we will discuss (in fact, this is obvious)
how to get an additional insurance:
not to lose much as compared to the current value of the underlying security.

Condition (\ref{eq:SLA}) implies $\liminf_{y\to\infty}F(y)/y=0$
(and even $\lim_{y\to\infty}F(y)/y=0$,
as we will show in Lemma~\ref{lem:to-zero} below).
Therefore, $\K_t/X_t$ may be very small for some $t$
even if $\K_t\ge F(X^*_t)$ holds.
A simple way
to insure against this possibility
is to hold $c\in(0,1)$ units of $X$ (assuming $X_0=1$)
and to invest $1-c$ into a strategy ensuring $\K_t\ge F(X^*_t)$.
The following corollary says that it leads to an optimal result.

\begin{proposition}\label{prop:insurance}
  Let $c\ge0$ and $F:[1,\infty)\to[0,\infty)$ be an increasing function.
  Investor has a strategy ensuring
  \begin{equation}\label{eq:insurance}
    \K_t\ge cX_t+F(X_t^*)
  \end{equation}
  if and only if $c$ and $F$ satisfy
  \begin{equation}\label{eq:condition}
    \int_1^{\infty} \frac{F(y)}{y^2} \dd y \le 1-c.
  \end{equation}
\end{proposition}
\begin{proof}
  Suppose (\ref{eq:condition}) is satisfied;
  in particular, $c\in[0,1]$.
  The case $c=1$ is trivial, so we assume $c<1$.
  Using $c+(1-c)p'_t$ as Investor's strategy,
  where $p'_t$ are Investor's moves
  guaranteeing $\K_t\ge\frac{1}{1-c}F(X^*_t)$ (cf.\ Proposition~\ref{prop:SLA}),
  we can see that Investor can guarantee (\ref{eq:insurance}).

  The rest of the proof is similar to the second part
  of the proof of Proposition~\ref{prop:SLA},
  and is again informal, for the same reasons.
  Suppose (\ref{eq:insurance}) is satisfied;
  our goal is to demonstrate (\ref{eq:condition}).
  Without loss of generality,
  assume that $F$ is left-continuous.
  Again replacing the discrete time parameter $t\in\{0,1,\ldots\}$ by $t\in[0,\infty)$,
  assuming that $X_t$ is the trajectory of Brownian motion
  started from 1 and stopped when it hits 0,
  and taking the expected value of both sides of (\ref{eq:insurance}),
  we obtain $\Expect F(X^*_t)\le 1-c$;
  by the monotone convergence theorem,
  letting $t\to\infty$ gives $\Expect F(X^*_{\infty})\le 1-c$,
  i.e., (\ref{eq:condition}).
\end{proof}

In fact, the guarantee (\ref{eq:insurance}), and an even stronger guarantee,
can be extracted directly from Equation (\ref{eq:guarantee})
in the previous section.
If we
do not discard the term
$\int_{(X_t^*,\infty]} \K_t^{(u)}P(\dd u)$
in (\ref{eq:guarantee}), we will obtain
\begin{equation}\label{eq:stronger-guarantee}
  \K_t
  \ge
  P((X_t^*,\infty])
  X_t
  +
  F(X_t^*).
\end{equation}
The coefficient $P((X_t^*,\infty])$ in front of $X_t$
shrinks to $c:=P(\{\infty\})$ as $X_t^*\uparrow\infty$,
and the function $F$ in (\ref{eq:stronger-guarantee})
satisfies (\ref{eq:condition}).
Therefore, (\ref{eq:stronger-guarantee}) is stronger than (\ref{eq:insurance}).
This does not contradict the part ``only if'' of Proposition~\ref{prop:insurance},
which does not say that
(\ref{eq:insurance}) cannot be improved;
it only says that the improvement will not be significant enough
to decrease the coefficient in front of $X_t$.

The purpose of the next two sections will be to show that (\ref{eq:stronger-guarantee})
is all we can get even in the situation
when we allow an arbitrary dependence of the right-hand side on $X^*_t$ and $X_t$.

According to (\ref{eq:class-1}) and (\ref{eq:insurance}),
Investor can guarantee
\begin{equation}\label{eq:alpha-insurance}
  \K_t
  \ge
  cX_t
  +
  (1-c)\alpha
  (X^*_t)^{1-\alpha}
\end{equation}
for any constants $c\in[0,1]$ and $\alpha\in(0,1)$.
In Appendix~\ref{app:analytic} we will see that using (\ref{eq:stronger-guarantee})
allows us to improve~(\ref{eq:alpha-insurance}) to
\begin{equation}\label{eq:improvement}
  \K_t
  \ge
  cX_t
  +
  (1-c)\alpha(X^*_t)^{1-\alpha}
  +
  (1-c)(1-\alpha)(X_t^*)^{-\alpha}X_t.
\end{equation}

\section{Insuring against loss of capital, III}
\label{sec:capital-3}

In this section we consider more general lookback adjusters,
those that depend on both $X^*_t$ and $X_t$.
A positive function $F(X^*,X)$,
where $X^*$ ranges over $[1,\infty)$ and $X$ over $[0,X^*]$,
is an \emph{LA} if there exists a strategy for Investor
that guarantees $\K_t\ge F(X^*_t,X_t)$ for all $t$.
An LA $F$ \emph{dominates} an LA $G$
if $F(X^*,X)\ge G(X^*,X)$ for all $X^*\in[1,\infty)$ and $X\in[0,X^*]$.
We say that $F$ \emph{strictly dominates} $G$
if $F$ dominates $G$
and $F(X^*,X)>G(X^*,X)$ for some $X^*\in[1,\infty)$ and $X\in[0,X^*]$.
An LA is an \emph{ALA} if it is not strictly dominated by any LA.

Remember that by $f\rd$ we mean the right derivative of $f$;
in particular, $F^=\rd$ is the right derivative of $F^=$.
\begin{theorem}\label{thm:general}
  Every LA is dominated by an ALA.
  A positive function $F(X^*,X)$
  with domain $X^*\in[1,\infty)$ and $X\in[0,X^*]$
  is an ALA
  if and only if the following two conditions are satisfied:
  \begin{itemize}
  \item
    the function
    \begin{equation}\label{eq:ALA2spine}
      F^=(X^*):=F(X^*,X^*),
      \quad
      X^*\in[1,\infty),
    \end{equation}
    is increasing, concave, and satisfies $F^=(1)=1$ and $F^=\rd(1)\le1$;
  \item
    for each $X^*\in[1,\infty)$,
    the function $F(X^*,X)$ is linear in $X$
    and its slope is equal to the right derivative of $F^=$ at the point $X^*$.
  \end{itemize}
\end{theorem}

Theorem~\ref{thm:general} will be deduced from three lemmas.
The function $F^=:[1,\infty)\to[0,\infty)$ defined by (\ref{eq:ALA2spine})
will be called the \emph{spine} of an ALA $F(X^*,X)$.

By a \emph{situation} we mean any sequence
$\sigma=(X_1,\ldots,X_t)$ of Market's moves;
$\Box$ stands for the empty situation.
We use the notation $\X(\sigma)$ for the last move $X_t$ of Market
and the notation $\Xstar(\sigma)$ for the highest price $\max_{s=0,\ldots,t}X_s$
of the security so far,
setting $\X(\Box)=\Xstar(\Box):=1$.
If $\Pi$ is a strategy for Investor,
$\K^{\Pi}(\sigma)$ is defined as Investor's capital $\K_t$ in the situation $\sigma$
when Investor follows $\Pi$.
Formally, a \emph{strategy for Investor}
(also called a \emph{trading strategy}) is defined as a function
$\Pi:\Sigma\to\bbbr$,
where $\Sigma$ is the set of all situations,
and
$$
  \K^{\Pi}(X_1,\ldots,X_t)
  :=
  1
  +
  \sum_{s=1}^{t}
  p_s(X_s-X_{s-1}),
$$
where $p_s:=\Pi(X_1,\ldots,X_{s-1})$.

\begin{lemma}\label{lem:direct}
  If a positive function $F(X^*,X)$, $X^*\in[1,\infty)$, $X\in[0,X^*]$,
  satisfies the two conditions in the statement of Theorem~\ref{thm:general},
  it is an LA.
\end{lemma}
\begin{proof}
  The following trading strategy witnesses that $F$ is an LA:
  at any time $t$, take the position
  $p_t:=F^=\rd(X^*_{t-1})$.
  (When we say that a trading strategy $\Pi$ witnesses that $F$ is an LA
  we mean that $\K^{\Pi}(\sigma)\ge F(\Xstar(\sigma),\X(\sigma))$
  for all situations $\sigma$.)
\end{proof}

\begin{lemma}\label{lem:dominated}
  Every LA is dominated by a function
  that satisfies the two conditions
  in the statement of Theorem~\ref{thm:general}.
\end{lemma}
\begin{proof}
  Let $F(X^*,X)$ be an LA.
  Choose a trading strategy $\Pi$ that witnesses that $F$ is an LA.
  Notice that $\Pi$'s moves $p_t$ are always positive, $p_t\ge0$:
  indeed, if $p_t<0$, Market can make $\K^{\Pi}$
  negative by choosing large enough $X_t$.

  Define $F_1(X^*,X)$ as the infimum of $\K^{\Pi}(\sigma)$
  over the situations $\sigma$ such that $\Xstar(\sigma)=X^*$ and $\X(\sigma)=X$.
  It is clear that $F_1$ is finite
  (in particular, $F_1(X^*,X)\le1+\Pi(\Box)(X^*-1)\le X^*$)
  and $F_1$ dominates $F$.
  Set $F_1^=(X):=F_1(X,X)$, $X\in[1,\infty)$.
  Let $F_2^=$ be the smallest concave increasing function that dominates $F_1^=$
  (in other words, $F_2^=$ is the lower envelope
  of the straight lines with positive slopes lying above the graph of $F_1^=$),
  and set $F_2(X^*,X):=F^=_2(X^*)+(F^=_2)\rd(X^*)(X-X^*)$,
  where $X^*\in[1,\infty)$ and $X\in[0,X^*]$.

  First we check that $F_2$ dominates $F_1$.
  Suppose it does not.
  There exist $X\in[0,\infty)$ and $X^*\in[1,\infty)$
  such that $X<X^*$ and the point $A:=(X,F_1(X^*,X))$
  lies strictly above the straight line $L_2$ passing through $B:=(X^*,F_2^=(X^*))$
  and having slope $(F_2^=)\rd(X^*)$.
  Let $L_1$ be the straight line passing through the points $A$ and $B$;
  the slope of $L_1$ is strictly less than the slope of $L_2$.
  Consider two cases:
  \begin{description}
  \item[The case $F_2^=(X^*)=F^=_1(X^*)$.]
    The graph of $F_2^=$ is below $L_2$;
    therefore, by the definition of $F^=_2$,
    the graph of $F_1^=$ is also below $L_2$.
    Consider two possibilities:
    \begin{itemize}
    \item
      If the graph of $F_1^=$ does not contain any points
      in the interior of the space between $L_1$ and $L_2$
      to the right of $B$,
      then the graph of $F^=_1$ is below both $L_2$ and $L_1$,
      and therefore, the graph of $F^=_2$ is below both $L_2$ and $L_1$.
      But we know that the graph of $F^=_2$ cannot be below $L_1$
      to the right of $B$.
    \item
      Suppose the graph of $F_1^=$ contains some points
      in the interior of the space between $L_1$ and $L_2$
      to the right of $B$,
      and let $C:=(X',F_1^=(X'))$ be such a point.
      Then $B$ is strictly below $[A,C]$.
      By the definition of $F_1^=$,
      there is a situation $\sigma$ such that
      $\Xstar(\sigma)=\X(\sigma)=X^*$ and the point $(X^*,\K^{\Pi}(\sigma))$
      lies strictly below the segment $[A,C]$ connecting the points
      $A=(X,F_1(X^*,X))$ and $C=(X',F_1(X',X'))$.
      It is clear that regardless of $\Pi(\sigma)$,
      in the situation $\sigma$ Market can choose the next move in such a way
      as to violate $\K^{\Pi}\ge F_1(\Xstar,\X)$.
    \end{itemize}
  \item[The case $F_2^=(X^*)>F^=_1(X^*)$.]
    We consider two possibilities:
    \begin{itemize}
    \item
      If $(F^=_2)\rd(X^*)=0$,
      the slope of $L_1$ is strictly negative, which is impossible:
      by the definition of $F_1^=$ there is a situation $\sigma$
      such that $\Xstar(\sigma)=\X(\sigma)=X^*$
      and $\K^{\Pi}(\sigma)<F_2^=(X^*)<F_1(X^*,X)$;
      since $\Pi(\sigma)\ge0$,
      Market can violate $\K^{\Pi}\ge F_1(\Xstar,\X)$ by choosing $X$
      as the next move.
    \item
      Now suppose $(F^=_2)\rd(X^*)>0$.
      Notice that the function $F_2^=$ is affine (and its graph coincides with $L_2$)
      to the right of $X^*$ in a neighbourhood of $X^*$.
      There are $X'\le X^*$ and $X''>X^*$ such that the segment $[C',C'']$,
      where $C':=(X',F^=_1(X'))$ and $C'':=(X'',F^=_1(X''))$,
      has a positive slope and lies strictly above $(X^*,F^=_1(X^*))$.
      For each $\epsilon>0$, we can choose such a segment
      $[C',C'']=[C'_{\epsilon},C''_{\epsilon}]$
      in such a way that it lies completely in the $\epsilon$-neighbourhood of $L_2$;
      and it is easy to see that the distance between $C''_{\epsilon}$ and $B$
      will stay bounded away from 0 as $\epsilon\to0$.
      This implies that $B$ will lie strictly below the segment $[A,C''_{\epsilon}]$
      for a small enough $\epsilon$.
      Therefore, $(X^*,F_1^=(X^*))$ will lie
      strictly below the segment $[A,C''_{\epsilon}]$.
      By the definition of $F_1^=$,
      there is a situation $\sigma$ such that
      $\Xstar(\sigma)=\X(\sigma)=X^*$ and the point $(X^*,\K^{\Pi}(\sigma))$
      lies strictly below the segment connecting the points
      $A=(X,F_1(X^*,X))$ and $C''_{\epsilon}=(X'',F_1(X'',X''))$,
      for some $X''>X^*$.
      Regardless of $\Pi(\sigma)$,
      in the situation $\sigma$ Market can choose the next move in such a way
      as to violate $\K^{\Pi}\ge F_1(\Xstar,\X)$.
    \end{itemize}
  \end{description}
  We can see that all possibilities lead to contradictions,
  which shows that $F_2$ indeed dominates $F_1$
  and, therefore, dominates $F$.

  The function $F_2$ satisfies all properties listed in the two conditions
  in the statement of Theorem~\ref{thm:general}
  possibly except $F_2^=(1)=1$ and $(F_2^=)\rd(1)\le1$.
  It remains to prove $F_2^=(1)\le1$ and $(F_2^=)\rd(1)\le1$:
  indeed, in this case $F_2$ will be dominated by a function
  satisfying the two conditions.
  Since $F_1^=(X)\le1+\Pi(\Box)(X-1)$ for all $X\ge1$,
  we have $F_2^=(1)\le1$.
  And if $(F_2^=)\rd(1)>1$,
  we would have $F(1,0)\le F_2(1,0)=F_2^=(1)-(F^=_2)\rd(1)<0$.
\end{proof}

\begin{lemma}\label{lem:not-comparable}
  If positive functions $F_1(X^*,X)$ and $F_2(X^*,X)$,
  $X^*\in[1,\infty)$, $X\in[0,X^*]$,
  satisfy the two conditions in the statement of Theorem~\ref{thm:general}
  and $F_1\le F_2$,
  then $F_1=F_2$.
\end{lemma}
\begin{proof}
  %
  Suppose $F_1$ and $F_2$ satisfy the conditions in the statement of the lemma
  but $F_1\ne F_2$.
  Since the functions satisfying the two conditions
  in 
  Theorem~\ref{thm:general}
  are determined by their spines,
  $F_1^=$ and $F_2^=$ must be different.
  Set $F(X^*,X):=F_2(X^*,X)-F_1(X^*,X)\ge0$
  and $F^=(X):=F(X,X)\ge0$.
  Suppose $F^=(X)>0$ for some $X\in[1,\infty)$;
  we fix such $X$ and will arrive at a contradiction.
  Let $X^*$ be a point in $[1,X]$ with the highest value of $F^=\rd$
  to within a small $\epsilon>0$;
  in particular, $F^=\rd(X^*)>0$.
  Since $F^=$ is absolutely continuous, we have:
  \begin{multline*}
    F^=(X^*)
    =
    \int_{[1,X^*]}
    F^=\rd(x)
    \dd x\\
    \le
    \int_{[1,X^*]}
    (F^=\rd(X^*)+\epsilon)
    \dd x
    =
    (X^*-1)(F^=\rd(X^*)+\epsilon).
  \end{multline*}
  Since
  \begin{align*}
    F(X^*,X)
    &=
    F(X^*,X^*)
    +
    F^=\rd(X^*)
    (X-X^*)\\
    &\le
    (X^*-1)(F^=\rd(X^*)+\epsilon)
    +
    F^=\rd(X^*)
    (X-X^*)\\
    &=
    -F^=\rd(X^*)+(X^*-1)\epsilon+F\rd^=(X^*)X,
  \end{align*}
  $F(X^*,0)$ will be strictly negative for $\epsilon$ small enough;
  this contradicts our assumption $F_1\le F_2$.
\end{proof}

\begin{proof}[Proof of Theorem \ref{thm:general}]
  In view of Lemma~\ref{lem:dominated},
  it suffices to prove that any ALA satisfies the two conditions
  in the statement of the theorem
  and that any function satisfying the two conditions
  is an ALA.

  Suppose $F$ is an ALA.
  By Lemmas \ref{lem:dominated} and \ref{lem:direct},
  it is dominated by an LA $F'$ satisfying the two conditions.
  By admissibility, $F=F'$.

  Suppose a function $F$ satisfies the two conditions.
  By Lemma~\ref{lem:direct}, $F$ is an LA.
  By Lemma~\ref{lem:dominated},
  it suffices to check that $F$ is not strictly dominated
  by a function satisfying the two conditions.
  It remains to apply Lemma~\ref{lem:not-comparable}.
\end{proof}

\section{Various connections}
\label{sec:connections}

Figure~\ref{fig:frame} provides a visual frame
for the relationships we discuss in this section
and elsewhere in this article.
ALAs are characterized by the two conditions
in Theorem \ref{thm:general}.
By a ``scaled ASLA'' we mean a function of the form $cF$,
where $c\in[0,1]$ and $F$ is an ASLA;
more fully, such functions may be called \emph{scaled down ASLAs}.
These are increasing right-continuous functions $F$
satisfying (\ref{eq:SLA}).
A \emph{spine} is a function that can be represented
as the spine of some ALA;
such functions are characterized by the first condition
in Theorem~\ref{thm:general}.
A ``measure'' stands for a probability measure on $[0,\infty]$.
We can see that the notions in all four vertices
of the square in Figure~\ref{fig:frame}
have simple analytic characterizations.

\begin{figure}
  \begin{center}
    \input{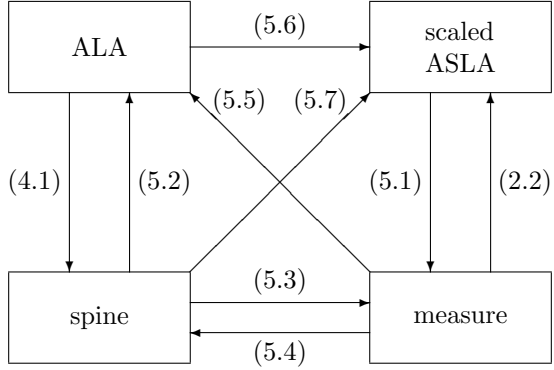}
  \end{center}
\caption{Some relationships between ALAs
  (functions satisfying the two conditions in Theorem~\ref{thm:general}),
  spines
  (concave increasing functions $F:[1,\infty)\to[0,\infty)$
  such that $F(1)=1$ and $F\rd(1)\le1$),
  probability measures on $[1,\infty]$,
  and scaled down ASLAs
  (right-continuous increasing functions $F:[1,\infty)\to[0,\infty)$
  satisfying $\int_1^{\infty}F(y)y^{-2}\dd y\le1$).\label{fig:frame}}
\end{figure}

The arrows in Figure~\ref{fig:frame} represent
various connections between the four notions;
they are labelled by the equations expressing those connections.
Each of the four sides of the square in Figure~\ref{fig:frame}
represents a bijective mapping between the sets of objects
in the adjacent vertices of the square.
The first such bijective mapping was introduced in Section~\ref{sec:capital-1};
it corresponds to the right side of the square.
Given a probability measure $P$ on $[1,\infty]$,
we define the corresponding scaled ASLA $F$ by (\ref{eq:measure2SLA}).
As can be seen from the proof of Lemma~\ref{lem:ASLA},
$P$ is uniquely determined by $F$,
and the expression of the restriction of $P$ to $[1,\infty)$
in terms of $F(X)$ is given there as
\begin{equation}\label{eq:SLA2measure}
  Q([1,y]) := F(y),
  \enspace
  y\in[1,\infty);
  \quad
  P(\dd u) := (1/u) Q(\dd u);
\end{equation}
$P(\{\infty\})$ is then determined uniquely as $1-P([1,\infty))$.

Another easy side of the square is the left one,
considered in Section~\ref{sec:capital-3}.
The spine $F^=$ is just the diagonal (\ref{eq:ALA2spine})
of the corresponding ALA $F$.
According to the second condition in Theorem~\ref{thm:general},
the expression of an ALA $F$ via its spine $F^=$ is
\begin{equation}\label{eq:spine2ALA}
  F(X^*,X)
  =
  F^=(X^*)
  +
  F^=\rd(X^*)
  (X-X^*).
\end{equation}

Next we consider the bottom side of the square.
The following lemma establishes a bijection
between the spines and the probability measures on $[1,\infty]$;
it uses (in the definition (\ref{eq:spine2measure}))
the obvious right-continuity of $F^=\rd$ for a spine $F^=$.
\begin{lemma}\label{lem:bottom-side}
  Let $F^=$ be a spine.
  Define a probability measure $P$ on $[1,\infty]$ by
  setting
  \begin{equation}\label{eq:spine2measure}
    P((X,\infty])
    :=
    F^=\rd(X),
    \quad
    X\in[1,\infty).
  \end{equation}
  Then
  \begin{equation}\label{eq:measure2spine}
    F^=(X)
    =
    \int_{[1,X]}u P(\dd u)
    +
    X
    P((X,\infty])
  \end{equation}
  for all $X\in[1,\infty)$.
  Vice versa,
  if $P$ is a probability measure on $[1,\infty]$,
  the function $F^=$ defined by (\ref{eq:measure2spine}) is a spine
  and satisfies (\ref{eq:spine2measure}).
\end{lemma}
\begin{proof}
  Let $F^=$ be a spine
  and a probability measure $P$ on $[1,\infty]$
  be defined by (\ref{eq:spine2measure}).
  Using integration by parts for the Lebesgue--Stiltjes integral
  (see, e.g., \cite{folland:1999}, Theorem~3.36),
  we obtain:
  \begin{multline*}
    \int_{[1,X]} u P(\dd u)
    =
    P(\{1\})
    +
    \int_{(1,X]} u P(\dd u)
    =
    1 - F^=\rd(1)
    -
    \int_{(1,X]} u \dd F\rd^=(u)\\
    =
    1 - F^=\rd(1)
    -
    XF^=\rd(X)
    +
    F^=\rd(1)
    +
    \int_{(1,X]} F\rd^=(u) \dd u\\
    =
    1
    -
    XF^=\rd(X)
    +
    F^=(X) - F^=(1)
    =
    F^=(X) - X P((X,\infty]).
  \end{multline*}
  The equality between the two extreme terms of this chain
  is equivalent to (\ref{eq:measure2spine}).

  We can see that the relations (\ref{eq:spine2measure}) and (\ref{eq:measure2spine})
  establish a bijection between the spines
  and a subset of probability measures on $[1,\infty]$.
  Now let $P$ be any probability measure on $[1,\infty]$
  and define $F^=:[1,\infty)\to[0,\infty)$ by $F^=(1):=1$
  and the equality
  $
    F^=\rd(X)
    =
    P((X,\infty])
  $,
  $X\in[1,\infty)$
  (cf.\ (\ref{eq:spine2measure})).
  Namely, set $F^=(X):=1+\int_{[1,X]}f(x)\dd x$,
  where $f:[1,\infty)\to[0,\infty)$ is the right-continuous decreasing function
  defined by $f(x):=P((x,\infty])$.
  It is easy to see that $F^=$ is a spine,
  and the argument of the previous paragraph shows
  that it satisfies (\ref{eq:measure2spine})
  (which can be taken as the definition of $F^=$).
  This completes the proof that (\ref{eq:spine2measure}) and (\ref{eq:measure2spine})
  establish a bijection between the spines
  and the probability measures on $[0,\infty]$.
\end{proof}

We have established the three bijections corresponding
to the right, left, and bottom sides of the square
in Figure~\ref{fig:frame}.
That figure also contains three shortcuts:
the top side and the diagonals of the square;
these are compositions of bijections
and so are bijections themselves.
(This structure of the diagram,
three basic bijections and three shortcuts,
makes sure that it ``commutes'',
in the terminology of category theory.)

First,
combining (\ref{eq:spine2ALA}), (\ref{eq:measure2spine}), and (\ref{eq:spine2measure}),
we obtain an expression of an ALA $F$
in terms of the corresponding measure $P$ on $[1,\infty]$:
\begin{align}
  F(X^*,X)
  &=
  F^=(X^*)
  +
  F^=\rd(X^*)
  (X-X^*)\notag\\
  &=
  \int_{[1,X^*]}u P(\dd u)
  +
  X^*
  P((X^*,\infty])
  +
  P((X^*,\infty])
  (X-X^*)\notag\\
  &=
  \int_{[1,X^*]}u P(\dd u)
  +
  X
  P((X^*,\infty])
  \label{eq:measure2ALA}
\end{align}
(cf.\ (\ref{eq:stronger-guarantee}) and (\ref{eq:measure2SLA})).

Second,
since the scaled ASLA corresponding to a probability measure $P$ on $[1,\infty]$
is (\ref{eq:measure2SLA})
and the ALA corresponding to $P$ is (\ref{eq:measure2ALA}),
we can see that the composition of
(\ref{eq:ALA2spine}), (\ref{eq:spine2measure}), and (\ref{eq:measure2SLA})
is the function
\begin{equation}\label{eq:ALA2SLA}
  F'(X^*):=F(X^*,0),
  \quad
  X^*\in[1,\infty),
\end{equation}
mapping each ALA $F$ to the corresponding scaled ASLA $F'$.

Third,
combining (\ref{eq:ALA2SLA}) and (\ref{eq:spine2ALA}),
we obtain an expression of the scaled ASLA in terms of the spine:
\begin{equation}\label{eq:spine2SLA}
  F'(X^*)
  =
  F^=(X^*)-F^=\rd(X^*)X^*;
\end{equation}
we can see that $F'(X)$ as a function of $-F^=\rd(X)$
is, essentially, the Legendre transformation of $-F^=(X)$.

The argument leading to (\ref{eq:ALA2SLA}) is important enough
to state its conclusion formally:
\begin{corollary}\label{cor:ALA-ASLA}
  Suppose $F(X^*,X)$ is an ALA.
  Then $F(X^*):=F(X^*,0)$ is a scaled ASLA.
  If, furthermore, $F^=\rd(\infty)=0$,
  $F(X^*)$ is an ASLA.
  Vice versa,
  if $F(X^*)$ is a scaled ASLA,
  there exists a unique ALA $F(X^*,X)$
  such that $F(X^*)=F(X^*,0)$ for all $X^*$.
  If, furthermore, $F(X^*)$ is an ASLA,
  this ALA $F(X^*,X)$ will satisfy $F\rd^=(\infty)=0$.
\end{corollary}

\begin{remark}
  Let us check analytically the first statement in Corollary~\ref{cor:ALA-ASLA}:
  if $F(X^*,X)$ satisfies the two conditions
  in 
  Theorem~\ref{thm:general},
  then $F(X^*):=F(X^*,0)$ satisfies (\ref{eq:SLA}),
  and if, furthermore, $F^=\rd(\infty)=0$,
  then $F(X^*)$ satisfies (\ref{eq:ASLA}).
  Since
  $$
    (F^=(y)y^{-1})\rd
    =
    F^=\rd(y)y^{-1}
    -
    F^=(y)y^{-2}
    =
    -\frac{F(y,0)}{y^2}
    =
    -\frac{F(y)}{y^2},
  $$
  the absolute continuity of the function $F^=(y)y^{-1}$ over $[1,\infty)$
  gives
  \begin{multline*}
    \int_1^{\infty}
    \frac{F(y)}{y^2}
    \dd y
    =
    -\left[F^=(y)y^{-1}\right]_{y=1}^{\infty}
    =
    F^=(1) - \lim_{y\to\infty}\frac{F^=(y)}{y}\\
    =
    1 - \lim_{y\to\infty}F^=\rd(y)
    \le
    1,
  \end{multline*}
  and ``${}\le1$'' becomes ``${}=1$'' when $F^=\rd(\infty)=0$.
\end{remark}

In the proof of Lemma~\ref{lem:bottom-side}
we have used the following alternative expression of a spine
in terms of the corresponding probability measure on $[0,\infty]$:
\begin{equation*} 
  F^=(X)
  =
  1
  +
  \int_{[1,X]}
  P((x,\infty])
  \dd x.
  \tag{\ref{eq:measure2spine}$'$}
\end{equation*}
Using (\ref{eq:measure2spine}$'$) in place of (\ref{eq:measure2spine})
in the derivation of (\ref{eq:measure2ALA}),
we obtain an alternative expression
\begin{equation*} 
  F(X^*,X)
  =
  P([1,X^*])
  +
  \int_{[1,X^*]}
  P((x,X^*])
  \dd x
  +
  P((X^*,\infty])X
  \tag{\ref{eq:measure2ALA}$'$}
\end{equation*}
of an ALA in terms of the corresponding probability measure on $[0,\infty]$.
In combination with (\ref{eq:ALA2SLA}),
this gives an alternative expression
\begin{equation*}
  F'(X^*)
  =
  P([1,X^*])
  +
  \int_{[1,X^*]}
  P((x,X^*])
  \dd x
  \tag{\ref{eq:measure2SLA}$'$}
\end{equation*}
of a scaled ASLA in terms of the corresponding measure.

\subsection*{Generalizations of Propositions~\ref{prop:SLA}
  and~\ref{prop:insurance}}

Theorem~\ref{thm:general} allows us to generalize
Propositions~\ref{prop:SLA} and~\ref{prop:insurance}
by dropping the requirement that the function $F$ should be increasing.
First we generalize the notions of SLA and ASLA.
A function $F:[1,\infty)\to[0,\infty)$ is an \emph{SLA}
if there exists a strategy for Investor that guarantees
$\K_t\ge F(X_t^*)$ for all $t$
(there are no measurability requirements on $F$).
We say that an SLA $F$ \emph{dominates} another SLA $G$
if $F(y)\ge G(y)$ for all $y\in[1,\infty)$.
We say that $F$ \emph{strictly dominates} $G$
if $F$ dominates $G$ and $F(y)>G(y)$ for some $y\in[1,\infty)$.
An SLA is an \emph{ASLA} if it is not strictly dominated by any SLA.
We will use the adjective ``increasing'' to refer to SLAs and ASLAs
as defined in Section~\ref{sec:capital-1}.
(In fact, Corollary~\ref{cor:SLA} will show that all ASLAs
are automatically increasing.)

\begin{lemma}\label{lem:equivalence}
  A function $G(X^*)$ is an SLA
  if and only if it has the form $F(X^*,0)$ for some LA $F$.
\end{lemma}
\begin{proof}
  First suppose that $G(X^*)=F(X^*,0)$, $\forall X^*\in[1,\infty)$,
  for some LA $F$.
  There is an ALA $F'\ge F$ (Theorem \ref{thm:general}).
  Some trading strategy ensures $\K_t\ge F'(X^*_t,X_t)$,
  and since $F'(X^*,X)$ is increasing in $X\in[0,X^*]$,
  it therefore ensures $\K_t\ge F'(X^*_t,0)\ge F(X^*_t,0)=G(X^*_t)$.
  So $G$ is an SLA.

  Now suppose that $G$ is an SLA.
  Then $F(X^*,X):=G(X^*)$ is an LA such that $G(X^*)=F(X^*,0)$.
\end{proof}

\begin{corollary}\label{cor:SLA}
\begin{enumerate}
\item
  A function $F:[1,\infty)\to[0,\infty)$ is an SLA
  if and only if it satisfies
  \begin{equation}\label{eq:superSLA}
    \int_1^{\infty}
    \frac{F^*(y)}{y^2}
    \dd y
    \le
    1,
  \end{equation}
  where $F^*(y):=\sup_{x\in[1,y]}F(x)$.
\item
  Any SLA is dominated by an ASLA.
\item
  An SLA is an ASLA if and only if
  it is increasing, right-continuous, and satisfies (\ref{eq:ASLA}).
\end{enumerate}
\end{corollary}
\begin{proof}
  First we prove part~1.
  If (\ref{eq:superSLA}) is true,
  $F^*$ is an SLA and so, \emph{a fortiori},
  $F$ is an SLA as well.

  In the opposite direction,
  if $F$ is an SLA,
  $F(X^*)=F_1(X^*,0)$, $\forall X^*\in[1,\infty)$,
  for some LA $F_1$
  (see Lemma~\ref{lem:equivalence}).
  By Theorem \ref{thm:general},
  $F_1$ is dominated by an ALA $F_2$.
  The function $F_3(X^*):=F_2(X^*,0)$ of $X^*\in[1,\infty)$
  is an increasing SLA (by Corollary~\ref{cor:ALA-ASLA})
  that dominates $F$ and, therefore, $F^*$.
  Now (\ref{eq:superSLA}) follows from
  $
    \int_1^{\infty}
    F_3(y)/y^2
    \dd y
    \le
    1
  $.

  Part~3 is now obvious since, by part~1, ASLAs
  must be increasing functions.
  Part~2 follows from parts~1 and~3.
\end{proof}

\begin{corollary}\label{cor:c}
  Let $c\ge0$ and $F:[1,\infty)\to[0,\infty)$.
  Investor has a strategy ensuring (\ref{eq:insurance})
  if and only if $c$ and $F$ satisfy
  \begin{equation}\label{eq:superSLA-general}
    \int_1^{\infty} \frac{F^*(y)}{y^2} \dd y \le 1-c.
  \end{equation}
\end{corollary}
\begin{proof}
  If (\ref{eq:superSLA-general}) is satisfied,
  Investor can ensure (\ref{eq:insurance}) with $F$ replaced by $F^*$,
  and so can ensure (\ref{eq:insurance}) itself.

  In the opposite direction,
  suppose Investor can ensure (\ref{eq:insurance}).
  It means that the function $F_1(X^*,X):=cX+F(X^*)$ is an LA.
  Let $F_2$ be any ALA that dominates $F_1$.
  Represent $F_2$ in the measure form (\ref{eq:measure2ALA}):
  $F_2(X^*,X)=P((X^*,\infty])X+F_3(X^*)$,
  where $F_3(X^*)=\int_{[1,X^*]}uP(\dd u)$.
  Since $F_3(X^*)/X^*\to0$ as $X^*\to\infty$
  (see Lemma~\ref{lem:to-zero} below),
  we have
  $$
    P(\{\infty\})
    =
    \lim_{X^*\to\infty}\frac{F_2(X^*,X^*)}{X^*}
    \ge
    \lim_{X^*\to\infty}\frac{F_1(X^*,X^*)}{X^*}
    \ge
    c.
  $$
  And since
  $$
    F(X^*)=F_1(X^*,0)\le F_2(X^*,0)=F_3(X^*),
  $$
  $F_3$ is an increasing function that dominates $F$,
  thus dominating $F^*$.
  Therefore,
  $$
    \int_1^{\infty}
    \frac{F^*(y)}{y^2}
    \dd y
    \le
    \int_{[1,\infty)}
    \frac{F_3(y)}{y^2}
    \dd y
    =
    P([1,\infty))
    =
    1-P(\{\infty\})
    \le
    1-c
  $$
  (the first equality follows from Lemma~\ref{lem:ASLA}).
\end{proof}

The following lemma (in combination with Lemma~\ref{lem:ASLA})
was used in the proof of Corollary~\ref{cor:c}.
\begin{lemma}\label{lem:to-zero}
  If an increasing function $F:[1,\infty)\to[0,\infty)$ satisfies (\ref{eq:SLA}),
  $\lim_{y\to\infty}F(y)/y=0$.
\end{lemma}
\begin{proof}
  If $\int_1^{\infty} F(y)y^{-2}\dd y<\infty$ for increasing $F$,
  then $\int_c^{\infty} F(y)y^{-2}\dd y\to0$ as $c\to\infty$,
  and so $\int_c^\infty F(c)/y^{-2}\dd y=F(c)/c\to0$ as $c\to\infty$.
\end{proof}

\section{Trading algorithm}
\label{sec:algorithms}

In this short section we will give an explicit trading strategy
(already described briefly in the proof of Lemma~\ref{lem:direct})
ensuring $\K_t\ge F(X^*_t,X_t)$ for all $t$,
where $F$ is an ALA,
or $\K_t\ge F'(X^*_t)$ for all $t$,
where $F'$ is an ASLA,
in the notation of Protocol~\ref{prot:competitive-trading}.
This strategy can be given
in terms of either the corresponding spine $F^=$
(in the spirit of Section \ref{sec:capital-1})
or the corresponding probability measure $P$ on $[0,\infty]$
(in the spirit of Section \ref{sec:capital-3}).

If we would like to ensure that $\K_t\ge F(X^*_t,X_t)$
for some ALA $F$,
we can apply Algorithm~\ref{alg:explicit}
to the spine $F^=(X^*):=F(X^*,X^*)$ of $F$.

\begin{algorithm} 
  \caption{Ensuring $\K_t\ge F(X_t^*,X_t)$ or $\K_t\ge F'(X_t^*)$}
  \label{alg:explicit}
  \begin{algorithmic}
    \REQUIRE spine $F^=:[1,\infty)\to[0,\infty)$
    \STATE $X^*:=1$
    \FOR{$t=1,2,\dots$}
      \STATE hold $F^=\rd(X^*)$ units of $X$
      \STATE read $X_t$
      \STATE $\K_t := \K_{t-1} + F^=\rd(X^*) (X_t-X_{t-1})$
      \IF{$X_t>X^*$}
        \STATE $X^*:=X_t$
      \ENDIF
    \ENDFOR
  \end{algorithmic}
\end{algorithm}

If we would like to ensure that $\K_t\ge F'(X^*_t)$
for an ASLA $F'$,
we first need to find the spine $F^=$ corresponding to $F'$;
in other words, to find $F^=$ satisfying (\ref{eq:spine2SLA}).
This can be done by combining (\ref{eq:SLA2measure}) and (\ref{eq:measure2spine}).
After that we can apply Algorithm~\ref{alg:explicit}.

Alternatively,
we could use the probability measure $P$ on $[1,\infty)$
corresponding to $F$ or $F'$, respectively,
as the parameter of Algorithm~\ref{alg:explicit}:
the only difference would be that $F^=\rd(X^*)$
would be replaced by $P((X^*,\infty])$
(cf.\ (\ref{eq:spine2measure})).
This is exactly the trading strategy
used in the proof of Proposition~\ref{prop:SLA}:
see (\ref{eq:strategy}) and (\ref{eq:p}).

\section{Pricing adjusted American lookbacks}

In this section we will consider
a modified version of Protocol~\ref{prot:competitive-trading},
given as Protocol~\ref{prot:trading}.
Now Investor starts with initial capital $\K_0$ equal to $\alpha$, and
the security's initial price $X_0$
is not necessarily 1 but is chosen by Market.

\begin{protocol}
  \caption{Trading in a financial security}
  \label{prot:trading}
  \begin{algorithmic}
    \STATE $\K_0:=\alpha$
    \STATE Market announces $X_0\in[0,\infty)$
    \FOR{$t=1,2,\dots$}
      \STATE Investor announces $p_t\in\bbbr$
      \STATE Market announces $X_t\in[0,\infty)$
      \STATE $\K_t:=\K_{t-1}+p_t(X_t-X_{t-1})$
    \ENDFOR
  \end{algorithmic}
\end{protocol}

A \emph{situation} in Protocol~\ref{prot:trading}
is a non-empty sequence $\sigma=(X_0,X_1,\ldots,X_t)$ of Market's moves,
which now includes $X_0$.
We let $\Sigma$ stand for the set of all situations.
A \emph{strategy for Investor} (or \emph{trading strategy})
is a function $\Pi:\Sigma\to\bbbr$,
and
$$
  \K^{\alpha,\Pi}(X_0,X_1,\ldots,X_t)
  :=
  \alpha +
  \sum_{s=1}^{t}
  \Pi(X_0,\ldots,X_{s-1})
  (X_s-X_{s-1})
$$
is Investor's capital in a situation $(X_0,X_1,\ldots,X_t)$
when he follows $\Pi$ from initial capital $\alpha$.
A \emph{capital process} is a real-valued function on $\Sigma$
that can be represented in the form $\K^{\alpha,\Pi}$
for some $\alpha$ and $\Pi$.

Let $F:\Sigma\to\bbbr$.
The perpetual American option with payoff $F$ entitles its owner
to the payoff $F(X_0,X_1,\ldots,X_t)$
at the time $t\in\{0,1,\ldots\}$ of her choice.
The \emph{upper price} of (the American option with payoff) $F$
in a situation $\iota$ is defined as
\begin{equation}\label{eq:upper-expectation-1}
  \UpExpect(F\givn\iota)
  :=
  \inf
  \left\{
    \K(\iota)
    \st
    \K(\sigma)\ge F(\sigma), \forall\sigma\in\Sigma_{\iota}
  \right\},
\end{equation}
where $\K$ ranges over the capital processes
and $\Sigma_{\iota}$ stands for the set of all situations $\sigma$
such that $\iota$ is a prefix of $\sigma$.
Intuitively, $\UpExpect(F\givn\iota)$ is the price of a cheapest superhedge
for $F$ in the situation $\iota$.

Let $\Omega$ be the set of all infinite sequences
$X_0,X_1,X_2,\ldots$ of Market's moves,
and let $F:\Omega\to(-\infty,\infty]$.
The European option with maturity date $\infty$ and payoff $F$
entitles its owner
to the payoff $F(X_0,X_1,X_2,\ldots)$ at time $\infty$.
The \emph{upper price} of (the European option with maturity date $\infty$ and payoff) $F$
in a situation $\iota$ is defined as
\begin{multline}\label{eq:upper-expectation-2}
  \UpExpect(F\givn\iota)
  :=
  \inf
  \Bigl\{
    \K(\iota)
    \st
    \forall(X_0,X_1,X_2,\ldots)\in\Omega_{\iota}:\\
    \liminf_{t\to\infty}\K(X_0,X_1,\ldots,X_t)\ge F(X_0,X_1,X_2,\ldots)
  \Bigr\},
\end{multline}
where $\K$ ranges over the capital processes
and $\Omega_{\iota}$ is the set of all sequences in $\Omega$
containing $\iota$ as their prefix.

Using the notation $\UpExpect$ in this section
usually implies that the corresponding infimum
(see (\ref{eq:upper-expectation-1}) and (\ref{eq:upper-expectation-2}))
is attained;
the only exception is the second statement of Corollary~\ref{cor:pricing-2}.

As discussed in Section~\ref{sec:introduction},
the results of the previous sections can be recast as a study of the upper prices
of perpetual American options paying $G(X_t^*,X_t)$
for various functions $G$.
The following corollaries list some special cases,
complemented with simple statements about European options.

\begin{corollary}\label{cor:pricing-1}
  Let $G:[0,\infty)\to[0,\infty)$ be an increasing function
  and $X_0\in(0,\infty)$.
  The upper price in the situation $X_0$
  of the perpetual American option with payoff $G(X^*_t)$ is
  $
    X_0\int_{X_0}^{\infty}G(x)x^{-2}\dd x
  $.
  The upper price in the situation $X_0$
  of the European option paying $G(X^*_{\infty})$
  at $\infty$
  is also $X_0\int_{X_0}^{\infty}G(x)x^{-2}\dd x$.
\end{corollary}

\begin{corollary}\label{cor:pricing-2}
  Let $c\ge0$, $G:[0,\infty)\to[0,\infty)$ be an increasing function,
  and $X_0\in(0,\infty)$.
  The upper price in the situation $X_0$
  of the perpetual American option with payoff $cX_t+G(X^*_t)$
  is $cX_0+X_0\int_{X_0}^{\infty}G(x)x^{-2}\dd x$.
  The upper price in the situation $X_0$ of the European option
  paying $cX_{\infty}+G(X^*_{\infty})$ at time $\infty$,
  where $cX_{\infty}:=\infty$ when $\lim_{t\to\infty}X_t$ does not exist,
  is $cX_0+X_0\int_{X_0}^{\infty}G(x)x^{-2}\dd x$.
\end{corollary}
\begin{proof}
  The only statement going beyond the argument in Section~\ref{sec:introduction}
  is the one about European options;
  namely, we need to justify the convention
  $cX_{\infty}:=\infty$ when $\lim_{t\to\infty}X_t$ does not exist.
  By the argument in Doob's martingale convergence theorem
  (see, e.g., \cite{shafer/vovk:2001}, Lemma~4.5),
  there exists a strategy $\Pi$ for Investor
  such that $\K^{1,\Pi}$ is always positive
  and $\K^{1,\Pi}(X_1,\ldots,X_t)\to\infty$ as $t\to\infty$
  when $\lim_{t\to\infty}X_t$ does not exist.
  Finally, we can replace the initial capital 1 of $\K^{1,\Pi}$
  by an arbitrarily small $\epsilon>0$.
\end{proof}

\subsection*{Pricing at time $s>0$}

A natural question is what the upper price of the perpetual American option
with payoff $G(X^*_t)$
is at a time $s>0$.
The answer can be obtained by applying the formula
$X_0\int_{X_0}^{\infty}G(x)x^{-2}\dd x$
to the function $x\mapsto G(X^*_s\vee x)$ 
(where $u\vee v$ stands for $\max(u,v)$)
in place of $G(x)$
and to $X_s$ in place of $X_0$;
this gives $X_s\int_{X_s}^{\infty} G(X^*_s \vee x)x^{-2}\dd x$.
The same argument is also applicable to the corresponding European option.
We state this as the following corollary.

\begin{corollary}
  Let $G:[1,\infty)\to[0,\infty)$ be an increasing function.
  The upper price in a situation $(X_0,\ldots,X_s)$ such that $X_s>0$
  of the perpetual American option with payoff $G(X^*_t)$ is
  $
    X_s\int_{X_s}^{\infty} G(X^*_s \vee x)x^{-2}\dd x
  $,
  where $X^*_s:=\max_{i\le s}X_i$.
  The upper price in a situation $(X_0,\ldots,X_s)$, $X_s>0$,
  of the European option paying $G(X^*_{\infty})$
  at $\infty$
  is also
  $
    X_s\int_{X_s}^{\infty} G(X^*_s \vee x)x^{-2}\dd x
  $.
\end{corollary}

\subsection*{More general American lookbacks, I}

Let $F(X^*,X)$ be a positive function whose domain includes all $(X^*,X)$
with $X^*>0$ and $X\in[0,X^*]$.
In this subsection
we will discuss the upper price in a situation $X_0>0$ of the American option
paying $F(X^*_t,X_t)$ at a time $t$ of the owner's choice.
To do this,
we first notice that the formula (\ref{eq:spine2ALA})
for transition from a spine to the corresponding ALA
can be applied to any concave increasing function
with domain $[X_0,\infty)$.
Formally,
we define an operator $G\mapsto\overline{G}$
on the concave increasing functions $G:[X_0,\infty)\to\bbbr$
by
\begin{align} 
  \overline{G}(X^*,X)
  :=
  G(X^*)
  +
  G\rd(X^*)&
  (X-X^*),\\
  &X^*\in[X_0,\infty),
  \enspace
  X\in[0,X^*]
  \label{eq:domain}
\end{align}
(our notation does not reflect the dependence of this operator on $X_0$).

The upper price $\UpExpect(F\givn X_0)$ of the American option paying $F(X^*_t,X_t)$
can be determined in two steps:
\begin{itemize}
\item
  Let $H:[X_0,\infty)\to[0,\infty)$ be the smallest concave increasing function
  such that $\overline{H}\ge F$ in the domain (\ref{eq:domain}).
  (The function $H$ can be defined as the infimum of all concave increasing functions $G$
  satisfying $\overline{G}\ge F$;
  the inequality $\overline{H}\ge F$ then follows from Lemma \ref{lem:exists} below.
  If such $G$ do not exist, set $H:=\infty$ on $[X_0,\infty)$.)
\item
  The function $H$ determines $\UpExpect(F\givn X_0)$ via
  \begin{equation}\label{eq:recipe}
    \UpExpect(F\givn X_0)
    =
    H(X_0).
  \end{equation}
\end{itemize}
Given the initial capital $H(X_0)$ in the situation $X_0$,
the option's seller can meet his obligation by holding
$p_t:=H\rd(X^*_{t-1})$
units of $X$ at time $t$.
And Theorem~\ref{thm:general} implies that $H(X_0)$
is the smallest initial capital allowing the option's seller
to meet his obligation for sure.

\begin{lemma}\label{lem:exists}
  Let $X_0>0$ and $\{G^{\alpha}\st\alpha\in A\}$
  be an indexed set of positive concave increasing functions $G^{\alpha}(X^*,X)$,
  where $(X^*,X)$ ranges over the domain (\ref{eq:domain}).
  Then
  $$
    \overline{\inf_{\alpha\in A}G^{\alpha}}
    \ge
    \inf_{\alpha\in A}
    \overline{G^{\alpha}}.
  $$
\end{lemma}
\begin{proof}
  Let 
  $
    \inf_{\alpha\in A}
    \overline{G^{\alpha}}
    \ge
    F
  $,
  i.e.,
  $\overline{G^{\alpha}}\ge F$ for all $\alpha\in A$.
  Our goal is to prove $\overline{H}\ge F$,
  where $H:=\inf_{\alpha\in A}G^{\alpha}$.
  Fix an arbitrary $(X^*,X)$ in the domain (\ref{eq:domain}).
  Our goal reduces to proving $\overline{H}(X^*,X)\ge F(X^*,X)$.

  Suppose $\overline{H}(X^*,X)\ge F(X^*,X)$ is false, i.e.,
  $$
    H(X^*) + H\rd(X^*)(X-X^*)
    <
    F(X^*,X).
  $$
  Taking $\Delta>0$ small enough, we obtain
  $$
    H(X^*) + \frac{H(X^*+\Delta)-H(X^*)}{\Delta}(X-X^*)
    <
    F(X^*,X).
  $$
  Choosing $\alpha\in A$
  such that $G^{\alpha}(X^*)$ is close enough to $H(X^*)$,
  we obtain
  $$
    G^{\alpha}(X^*) + \frac{G^{\alpha}(X^*+\Delta)-G^{\alpha}(X^*)}{\Delta}(X-X^*)
    <
    F(X^*,X),
  $$
  which implies
  $$
    G^{\alpha}(X^*) + G^{\alpha}\rd(X^*)(X-X^*)
    <
    F(X^*,X),
  $$
  which contradicts our assumption $\overline{G^{\alpha}}\ge F$.
\end{proof}

\subsection*{More general American lookbacks, II}

The lookbacks paying $X_t^*$ at some time $t$
that have been our motivation in this article
are the most basic ones,
but several other kinds have been considered in literature.
According to the standard nomenclature,
the full name for the American option paying $X^*_t$
at time $t\in[0,\infty)$
is ``perpetual American lookback call option with fixed strike 0''.
Fixing a finite maturity date $T$ does not change much
(it does not change anything at all
in our probability-free framework in the case of continuous time;
we have chosen the discrete-time framework in this article
only for simplicity).

Let $G$ be a positive increasing function.
Replacing the strike 0 by $c>0$ will change the pricing formula
for adjusted American lookbacks:
it is easy to see that the upper price in a situation $X_0>0$ of the American option
paying $G((X^*_t-c)^+)$ is
$X_0\int_{X_0}^{\infty}G((x-c)^+)x^{-2}\dd x$.
The other popular kinds of American lookbacks are:
\begin{itemize}
\item
  the American lookback put option with fixed strike $c$,
  whose payoff is $(c-\min_{s\le t}X_s)^+$;
\item
  the American lookback call option with floating strike,
  whose payoff is $X_t-\min_{s\le t}X_s$;
\item
  the American lookback put option with floating strike,
  whose payoff is $X^*_t-X_t$.
\end{itemize}
The first two payoffs depend on $\min_{s\le t}X_s$,
and so the methods of this article are not applicable to them.
The adjusted version of the last one
can be easily dealt with by our methods:
applying the recipe (\ref{eq:recipe}) to $F(X^*_t,X_t):=G(X^*_t-X_t)$,
we obtain $\UpExpect(F\givn X_0)=\UpExpect(F'\givn X_0)$, where
$$
  F'(X^*,X)
  :=
  F(X^*,0)
  =
  G(X^*).
$$
(Indeed, since $\overline{H}(X^*,X)$ is increasing in $X$
and $G(X^*-X)$ is decreasing in $X$,
the inequality $\overline{H}(X^*,X)\ge G(X^*-X)$ holds for all $(X^*,X)$
if and only if $\overline{H}(X^*,X)\ge G(X^*)$ holds for all $(X^*,X)$.)
Therefore,
by Corollary~\ref{cor:pricing-1},
$\UpExpect(F\givn X_0)=X_0\int_{X_0}^{\infty}G(x)x^{-2}\dd x$.
In other words, the term ``${}-X_t$'' in $G(X^*_t-X_t)$ does not help.

\section{Other connections with literature}
\label{sec:literature}

In addition to Hobson's approach mentioned in Section~\ref{sec:introduction},
this article's results have links with the recent probability-free version
\cite{\GTPxxviii}
(motivated by \cite{takeuchi/etal:2009})
of Dubins and Schwarz's \cite{dubins/schwarz:1965}
reduction of continuous martingales to Brownian motion
and with the Az\'ema--Yor solution \cite{azema/yor:1979}
to the Skorokhod embedding problem.

\subsection*{Risk-neutral probability measures}

In Section~\ref{sec:introduction},
we noticed that (\ref{eq:upper-price})
is the expected value of $G$ w.r.\ to the probability measure $Q_{X_0}$
on $[X_0,\infty)$ with density $X_0x^{-2}$.
In this somewhat informal subsection we will discuss the origins of $Q_{X_0}$.

A natural interpretation of $Q_{X_0}$ can be given
in the case of continuous time $[0,\infty)$
and a continuous price path $X_t$, $t\in[0,\infty)$.
For the details of the definition of capital processes, upper prices, etc.,
in continuous time,
see \cite{\GTPxxviii}.
It is easy to see that this article's results carry over to this continuous-time framework.
In particular, the upper price at time 0 of the European option
paying $G(X^*_{\infty})$ at time $\infty$,
where $G$ is a positive increasing function,
is equal to the expected value
$X_0\int_{X_0}^{\infty}G(x)x^{-2}\dd x$
with respect to the risk-neutral probability measure
$X_0x^{-2}\dd x$ on $[X_0,\infty)$.
In this section we will additionally assume that the function $G$
is bounded.

In the case of continuous price paths,
the emergence of the risk-neutral probability measure $X_0 x^{-2}\dd x$
on $[X_0,\infty)$
can be regarded as a corollary of the emergence of Brownian motion
discussed in \cite{\GTPxxviii}.
Indeed,
by Theorem~6.2 of \cite{\GTPxxviii},
the upper price of $G(X^*_{\infty})$ in the situation $X_0$
is equal to the expected value $\int G(X_0+\omega^*_{\tau})W(\dd\omega)$,
where $W$ is the Wiener measure on $\omega\in C([0,\infty))$
and $\tau:=\inf\{t\st X_0+\omega_t=0\}$.
In other words,
the upper price of $G(X^*_{\infty})$ in $X_0$
can be obtained by averaging $G$ with respect to the distribution $Q$
of the maximum of Brownian motion
started at $X_0$ and stopped when it hits 0.
The density of $Q$ is $X_0 x^{-2}$,
in agreement with this article's results;
indeed, the probability that Brownian motion started at $X_0$ hits level $x\ge X_0$,
before hitting 0
is $X_0/x$
(see, e.g., \cite{morters/peres:2010}, Theorem 2.49;
this follows from Brownian motion being a martingale);
therefore, the distribution function of $Q$ is $1-X_0/x$,
and its density is $X_0/x^2$.
This intuitive picture for the risk-neutral measure
was used in the informal parts of the proofs
of Propositions~\ref{prop:SLA} and~\ref{prop:insurance}.

It is easy to see that Brownian motion
can be replaced by any martingale in a wide class $\CCC$ of martingales.
By Dubins and Schwarz's classic result \cite{dubins/schwarz:1965},
each continuous martingale that is nowhere constant and unbounded almost surely
is a time-transformed Brownian motion;
therefore, we can include all such martingales in $\CCC$.

But it is clear that the class of allowable martingales is much wider;
e.g., in \cite{dawid/etal:2011} we used the martingale whose trajectories
are of the form
$$
  X_t
  =
  \begin{cases}
    1 & \text{if $t\le 1$}\\
    t & \text{if $1<t\le T$}\\
    0 & \text{otherwise},
  \end{cases}
$$
where $T\ge1$ depends on the trajectory
(we say ``the'' as this condition completely determines
the distribution of the martingale's trajectories).
The informal arguments in the proofs
of Propositions~\ref{prop:SLA} and~\ref{prop:insurance}
could have been based on this martingale rather than Brownian motion
(analogously to the proof of an analogous statement in \cite{dawid/etal:2011}:
cf.\ the end of the proof of Theorem~1 in \cite{dawid/etal:2011}).

In general, we can extend $\CCC$ by adding to it all right-continuous martingales $X_t$
that never make upward jumps when they are positive,
never make downward jumps from strictly positive to strictly negative values,
and such that
$\liminf_{t\to\infty}X_t\le0$ or $\limsup_{t\to\infty}X_t=\infty$ almost surely.
To see this,
use the standard martingale argument
given in \cite{morters/peres:2010}, Theorem 2.49.
(We assume that the first time when $X_t$ reaches or crosses some level
is a stopping time;
this will be the case for a reasonable choice of the definitions.)

\begin{remark}
  A very informal picture inspired by the use of improper priors in Bayesian statistics
  is that there is just one risk-neutral measure $Q$,
  with density $y^{-2}$ on $(0,\infty)$,
  and each probability distribution $Q_{X_0}$ for $X^*_{\infty}$
  is obtained from $Q$ by conditioning on the event $X^*_{\infty}\ge X_0$.
\end{remark}

\subsection*{ALAs and the Az\'ema--Yor solution to the Skorokhod embedding problem}

Let $X_t$, $t\in[0,\infty)$, be Brownian motion started at $0$.
Wald's lemmas (see, e.g., \cite{morters/peres:2010}, Theorems 2.44 and 2.48)
say that if $\tau$ is a stopping time with $\Expect\tau<\infty$,
we have $\Expect X_{\tau}=0$ and $\Expect X_{\tau}^2=\Expect\tau$.
The Skorokhod embedding problem goes in the opposite direction:
given a random variable $\xi$ with $\Expect\xi=0$ and $\Expect\xi^2<\infty$,
find a stopping time $\tau$ such that $X_{\tau}$ is distributed as $\xi$
and $\Expect\tau<\infty$ (i.e., $\Expect\tau=\Expect\xi^2$).
For a recent review of solutions to the Skorokhod embedding problem,
see \cite{obloj:2004}.

The most well-known solution to the Skorokhod embedding problem
is given by Az\'ema and Yor \cite{azema/yor:1979}.
It is based on the fact that if $f$ is a $C^1$ function,
the process $f(X^*_t)+(X_t-X_t^*)f\rd(X^*_t)$ is a local martingale.
(For a definitive generalization of this fact, see \cite{obloj:2006}.)
In other words, if $F^=$ is a $C^1$ function,
the process $F(X^*_t,X_t)$,
where $F$ is defined by (\ref{eq:spine2ALA}),
is a local martingale.
Therefore, the Az\'ema--Yor solution
is based on the notion of ALA
in which our requirements on a spine
are replaced by the requirement that a spine should be a $C^1$ function.

\section{Insuring against loss of evidence}
\label{sec:GTP}

In this section we will apply our results about insuring against loss of capital
to the problem of insuring against loss of evidence.
The latter problem was the topic of \cite{\GTPxxxiii}
in the standard framework of measure-theoretic probability;
we will consider the more general framework of game-theoretic probability.

In game-theoretic probability
(see, e.g., \cite{shafer/vovk:2001})
Sceptic tries to prove Forecaster wrong
by gambling against him:
the values of Sceptic's capital $\K_t$ measure the changing evidence against Forecaster.
We assume that Sceptic's initial capital is $\K_0=1$,
and that Sceptic is required to ensure that $\K_t\ge0$ at each time $t$.

Sceptic can lose as well as gain evidence.
At a time $t$ when $\K_t$ is large
Forecaster's performance looks poor,
but then $\K_i$ for some later time $i$ may be lower
and make Forecaster look better.
Our result (a simple corollary of the results of the previous sections)
will show that, for a modest cost, Sceptic can avoid losing too much evidence.

Suppose we exaggerate the evidence against Forecaster
by considering not the current value $\K_t$ of Sceptic's capital
but the greatest value so far:
$
  \K^*_t := \max_{s\le t} \K_s
$.
We will see that there are many functions $F:[1,\infty)\to[0,\infty)$
such that
\begin{enumerate}
\item
  $F(y)\to\infty$ as $y\to\infty$ almost as fast as $y$, and
\item
  Sceptic's moves can be modified on-line
  in such a way that the modified moves lead to capital
  \begin{equation}\label{eq:goal}
    \K'_t
    \ge
    F(\K_t^*),
    \quad
    t=1,2,\mathenddots
  \end{equation}
\end{enumerate}
If we are dissatisfied by the asymptotic character of the first of these two conditions,
which does not prevent $\K'_t/\K_t$ from becoming very small for some $t$,
we can compromise by putting a fraction $c\in(0,1)$ of the initial capital
on Sceptic's original moves and the remaining fraction $1-c$ on the modified moves,
thus obtaining capital $c\K_t+(1-c)\K'_t$ at each time $t$.
This way Sceptic may sacrifice a fraction $1-c$ of his capital
but gets extra insurance against losing evidence.

As we will see (in Corollary~\ref{cor:GTP}), the set of functions $F$
for which (\ref{eq:goal}) can be achieved
is exactly the set of all SLAs.

Our prediction protocol
(Protocol~\ref{prot:competitive-scepticism})
involves four players:
Forecaster, Sceptic, Rival Sceptic, and Reality.
The parameter of the protocol is a set $\XXX$,
from which Reality chooses her moves;
$\mathbf{E}$ is the set of all ``outer probability contents'' on $\XXX$
(to be defined shortly).
We always assume that $\XXX$ contains at least two distinct elements.
The reader who is not interested in the most general statement of our result
can interpret $\mathbf{E}$ as the set of all expectation functionals
$\EEE:f\mapsto\int f \dd P$,
$P$ being a probability measure on a fixed $\sigma$-algebra on $\XXX$;
in this case Sceptic and Rival Sceptic are required to output
functions that are measurable w.r.\ to that $\sigma$-algebra.

\begin{protocol} 
  \caption{Competitive scepticism}
  \label{prot:competitive-scepticism}
  \begin{algorithmic}
    \STATE $\K_0:=1$ and $\K'_0:=1$
    \FOR{$t=1,2,\dots$}
      \STATE Forecaster announces $\EEE_t\in\mathbf{E}$
      \STATE Sceptic announces $f_t\in[0,\infty]^{\XXX}$ such that $\EEE_t(f_t)\le\K_{t-1}$
      \STATE Rival Sceptic announces $f'_t\in[0,\infty]^{\XXX}$
        such that $\EEE_t(f'_t)\le\K'_{t-1}$
      \STATE Reality announces $x_t\in\XXX$
      \STATE $\K_t:=f_t(x_t)$ and $\K'_t:=f'_t(x_t)$
    \ENDFOR
  \end{algorithmic}
\end{protocol}

In general,
an \emph{outer probability content} on $\XXX$
is a function $\EEE:\bbbrbarXXX\to\bbbrbar$
(where $\bbbrbarXXX$ is the set of all functions $f:\XXX\to\bbbrbar$)
that satisfies the following four axioms:
\begin{enumerate}
\item 
  If $f,g\in\bbbrbarXXX$ and $f\le g$,
  then $\EEE(f)\le\EEE(g)$.
\item 
  If $f\in\bbbrbarXXX$ and $c\in(0,\infty)$, then $\EEE(cf)=c\EEE(f)$.
\item 
  If $f,g\in\bbbrbarXXX$, then $\EEE(f+g)\le\EEE(f)+\EEE(g)$.
\item 
  For each $c\in\bbbr$,
  $\EEE(c)=c$,
  where the $c$ in parentheses
  is the function in $\bbbrbarXXX$ that is identically equal to $c$.
\end{enumerate}
An axiom of $\sigma$-subadditivity on $[0,\infty]^{\XXX}$
is sometimes added to this list,
but we do not need it in this article.
(And it is surprising how rarely it is needed in general:
see, e.g., \cite{\GTPxxix}.)

\begin{remark}
  There is a dazzling array of terms that have been used in place
  of our ``outer probability contents''.
  In our terminology we follow \cite{hoffmann-jorgensen:1987} and \cite{\GTPxxix}.
  Upper previsions studied in the theory of imprecise probabilities
  (see, e.g., \cite{decooman/hermans:2008})
  are closely related to (but somewhat more restrictive than)
  outer probability contents.
  Coherent risk measures introduced in \cite{artzner/etal:1999}
  are essentially outer probability contents,
  but applied to $-f$ in place of $f$.
  A lot of different terms have been used by numerous authors
  developing \cite{artzner/etal:1999}.
\end{remark}

Protocol~\ref{prot:competitive-scepticism} describes a perfect-information game
in which Sceptic tries to discredit the outer probability contents $\EEE_t$
issued by Forecaster
as a faithful description of Reality's $x_t\in\XXX$.
On each round Sceptic and Rival Sceptic choose gambles $f_t$ and $f'_t$
on how $x_t$ is going to come out,
and their resulting capitals are $\K_t$ and $\K'_t$, respectively.
Discarding capital is allowed,
but Sceptic and Rival Sceptic are required to ensure
that $\K_t\ge0$ and $\K'_t\ge0$, respectively;
this is achieved by requiring that $f_t$ and $f'_t$ should be positive.

\begin{corollary}\label{cor:GTP}
  Let $F:[1,\infty)\to[0,\infty)$ be an increasing function.
  In Protocol~\ref{prot:competitive-scepticism},
  Rival Sceptic can ensure (\ref{eq:goal})
  if and only if $F$ is an SLA.
  More generally, let $c\in[0,1)$.
  Rival Sceptic can ensure
  \begin{equation}\label{eq:Rival-Sceptic-capital}
    \K'_t\ge c\K_t+F(\K_t^*),
    \enspace
    \forall t,
  \end{equation}
  if and only if $F/(1-c)$ is an SLA.
\end{corollary}

\noindent
The meaning of (\ref{eq:goal}) and (\ref{eq:Rival-Sceptic-capital})
when $\K_t^*=\infty$
is provided by the usual convention $F(\infty):=\lim_{y\to\infty}F(y)$.

\begin{proof}
  To establish the part ``if'',
  notice that Protocol~\ref{prot:competitive-scepticism} reduces
  to Protocol~\ref{prot:competitive-trading}
  (with Sceptic corresponding to Market
  and Rival Sceptic to Investor).
  In the latter, it is clear that any strategy for Investor
  ensuring (\ref{eq:insurance})
  always chooses $p_t\ge0$.
  Fix such a strategy $\Pi$.
  It can be used by Rival Sceptic in Protocol \ref{prot:competitive-scepticism}:
  if Sceptic's move on round $t$ is $f_t$
  and his capital at the beginning of the round is $\K_{t-1}<\infty$
  (so that $\EEE_t(f_t)\le\K_{t-1}$)
  and the strategy $\Pi$ recommends move $p_t$ for Investor,
  Rival Sceptic's move should be
  \begin{equation}\label{eq:Rival-Sceptic-move}
    f'_t
    :=
    \K'_{t-1} + p_t (f_t - \K_{t-1}).
  \end{equation}
  We will have both $\EEE_t(f'_t)\le\K'_{t-1}$
  and $\K'_t=\K'_{t-1}+p_t(\K_t-\K_{t-1})$.

  The case $\K_{t-1}=\infty$ has to be considered separately.
  Let $s\le t-1$ be the first time when $\K_s=\infty$.
  If $p_s>0$, we have $\K'_s=\infty$,
  and so we can set $f'_i:=\infty$ for all $i>s$;
  in particular, $\K'_t=\infty$.
  If $p_s=0$, we have $c=0$ and $\K'_{s-1}\ge F(\infty)$;
  therefore, (\ref{eq:Rival-Sceptic-capital}) will hold
  if we set $f'_i:=0$ for all $i\ge s$.

  The part ``only if'' follows from Protocol~\ref{prot:competitive-trading}
  being a special case of Protocol~\ref{prot:competitive-scepticism}.
  (One way to embed Protocol~\ref{prot:competitive-trading}
  into Protocol~\ref{prot:competitive-scepticism}
  is to set $\XXX:=[0,\infty)$ and make Forecaster output
  $$
    \EEE_t(f)
    :=
    \inf\{\K\st\exists p\in\bbbr\;\forall x\in\XXX:\K+p(x-X_{t-1})\ge f(x)\}
  $$
  on round $t$.)
\end{proof}

We refrain from giving a similar restatement of Theorem~\ref{thm:general}.

It is easy to see that Algorithm~\ref{alg:explicit}
is applicable not only in the financial context of Section~\ref{sec:algorithms}
but also in the context of Protocol~\ref{prot:competitive-scepticism}.
Namely, on round $t$ of Protocol~\ref{prot:competitive-scepticism}
Rival Sceptic should choose the move (\ref{eq:Rival-Sceptic-move}),
where $p_t$ is output by Algorithm~\ref{alg:explicit}.

In \cite{\GTPxxxiii} we use a simple method based on L\'evy's zero-one law
to prove a result similar to Corollary~\ref{cor:GTP}
that can be used for insuring against loss of evidence
in measure-theoretic probability and statistics.
As we explain there, 
the value $\K_t$ of the capital process is the dynamic version of Bayes factors,
and its running maximum $\K^*_t$ is the dynamic version of p-values;
SLAs transform inverse p-values into inverse Bayes factors.

\appendix
\makeatletter
  \renewcommand{\section}{\@startsection{section}{1}{0pt}%
  {-3.5ex plus -1ex minus -.2ex}{2.3ex plus.2ex}%
  {\normalfont\Large\bfseries\noindent Appendix~}}
\makeatother
\section{Details of the specific examples of ALAs and ASLAs}
\label{app:analytic}
\makeatletter
  \renewcommand{\section}{\@startsection{section}{1}{0pt}%
  {-3.5ex plus -1ex minus -.2ex}{2.3ex plus.2ex}%
  {\normalfont\Large\bfseries}}
\makeatother

In Section~\ref{sec:capital-1} we gave two examples of ASLAs,
(\ref{eq:class-1}) and (\ref{eq:class-2}).
In this appendix we will find the corresponding measures, spines, and ALAs
(cf.\ Figure \ref{fig:frame}).
It will be a good illustration of the absence
at the top of Figure~\ref{fig:frame}
of an arrow pointing to the left,
from ``scaled ASLA'' to ``ALA''.
To find the ALA corresponding to a given scaled ASLA,
we will have to move around the square via ``measure'' and ``spine''.

\subsection*{ASLAs and ALAs related to (\ref{eq:class-1})}

Let us first find the probability measure $P$ on $[1,\infty]$
corresponding to the ASLA $F$ defined by (\ref{eq:class-1}).
Using (\ref{eq:SLA2measure})
we find $Q([1,y])=\alpha y^{1-\alpha}$ for all $y\in[1,\infty)$,
and so $Q$ gives weight $\alpha$ to $1$
and has density $\alpha(1-\alpha)y^{-\alpha}$ over $(1,\infty)$.
Therefore, $P$ gives weight $\alpha$ to $1$
and has density $\alpha(1-\alpha)y^{-1-\alpha}$ over $(1,\infty)$;
it is clear that it gives weight $0$ to $\infty$.
Now we can find
\begin{equation}\label{eq:P}
  P((X,\infty])
  =
  \int_{X}^{\infty}
  \alpha(1-\alpha)y^{-1-\alpha}
  \dd y
  =
  (1-\alpha)
  X^{-\alpha}.
\end{equation}
We can see that the distribution function of the probability measure $P$ is
$
  P([1,X])
  =
  1
  -
  (1-\alpha)
  X^{-\alpha}
$,
$X\ge1$.

The spine corresponding to the $F$ defined by (\ref{eq:class-1})
has an even simpler expression:
using (\ref{eq:P}) and (\ref{eq:measure2spine}),
we obtain
$$
  F^=(X)
  =
  F(X) + X P((X,\infty])
  =
  \alpha X^{1-\alpha} + X (1-\alpha) X^{-\alpha}
  =
  X^{1-\alpha}.
$$

In Section~\ref{sec:capital-2} we implicitly considered the ALAs
corresponding to the probability measure
$P_c:=(1-c)P+c\delta_{\infty}$,
where $c\in[0,1]$
and $\delta_{\infty}$ is the probability measure on $[1,\infty]$
that is concentrated at $\infty$.
The corresponding spine is
$$
  F^=(X)
  =
  (1-c)X^{1-\alpha}+cX,
$$
and so, by (\ref{eq:spine2ALA}), the corresponding ALA is
\begin{align*}
  F(X^*,X)
  &=
  (1-c)(X^*)^{1-\alpha}+cX^*
  +
  \bigl(
    (1-c)
    (1-\alpha)
    (X^*)^{-\alpha}
    +
    c
  \bigr)
  (X-X^*)\\
  &=
  cX
  +
  (1-c)\alpha(X^*)^{1-\alpha}
  +
  (1-c)(1-\alpha)(X^*)^{-\alpha}X;
\end{align*}
cf.\ (\ref{eq:improvement}).

\subsection*{ASLAs and ALAs related to (\ref{eq:class-2})}

Let us now find the probability measure $P$ on $[1,\infty]$
and the spine $F^=$ corresponding to (\ref{eq:class-2}).
Since $Q([1,y])=\alpha(1+\alpha)^{\alpha}y\ln^{-1-\alpha}y$
when $y\in[e^{1+\alpha},\infty)$
and $Q([1,y])=0$ otherwise,
we obtain that $Q(\{e^{1+\alpha}\})=\frac{\alpha}{1+\alpha}e^{1+\alpha}$
and that over $(e^{1+\alpha},\infty)$ the measure $Q$
is absolutely continuous with density
$
  q(y)
  :=
  \alpha(1+\alpha)^{\alpha}\ln^{-1-\alpha}y
  -
  \alpha(1+\alpha)^{1+\alpha}\ln^{-2-\alpha}y
$.
Therefore, $P(\{e^{1+\alpha}\})=\frac{\alpha}{1+\alpha}$
and over $(e^{1+\alpha},\infty)$ the probability measure $P$
is absolutely continuous with density $q(y)/y$.
For any $X\ge e^{1+\alpha}$ we now obtain
\begin{align*}
  P((X,\infty))
  &=
  \alpha(1+\alpha)^{\alpha}
  \int_{X}^{\infty}
  \frac{\ln^{-1-\alpha}y}{y}
  \dd y
  -
  \alpha(1+\alpha)^{1+\alpha}
  \int_{X}^{\infty}
  \frac{\ln^{-2-\alpha}y}{y}
  \dd y\\
  &=
  (1+\alpha)^{\alpha}
  \ln^{-\alpha}X
  -
  \alpha(1+\alpha)^{\alpha}
  \ln^{-1-\alpha}X.
\end{align*}
Equation (\ref{eq:measure2spine}) now gives, for $X\ge e^{1+\alpha}$,
\begin{align*}
  F^=(X)
  &=
  \alpha(1+\alpha)^{\alpha}X\ln^{-1-\alpha}X
  +
  (1+\alpha)^{\alpha}
  X\ln^{-\alpha}X
  -
  \alpha(1+\alpha)^{\alpha}
  X\ln^{-1-\alpha}X\\
  &=
  (1+\alpha)^{\alpha}
  X\ln^{-\alpha}X.
\end{align*}
For $X<e^{1+\alpha}$,
the same equation gives $F^=(X)=X$.
Therefore,
$$
  F^=(X)
  =
  \begin{cases}
    (1+\alpha)^{\alpha}X\ln^{-\alpha}X & \text{if $X\ge e^{1+\alpha}$}\\
    X & \text{otherwise}.
  \end{cases}
$$
This function satisfies the first condition in the statement of Theorem~\ref{thm:general}
by definition;
it is also easy to check directly
(notice that $X\ln^{-\alpha}X$ is concave only over $(e^{1+\alpha},\infty)$).

\ifWP
\section*{Acknowledgements}

\ifFULL\bluebegin
  Comments by \emph{Finance and Stochastics} referees were very helpful.
\blueend\fi
\acknowledge
\fi

\end{document}